\RequirePackage[l2tabu, orthodox]{nag}

\documentclass[11pt,a4paper]{article}
\usepackage{fullpage}

\usepackage{hyperref}
\hypersetup{
  linktoc=all,
  colorlinks = false,
  citecolor=black,           
  unicode=true,              
  pdftoolbar=true,           
  pdfmenubar=true,           
  pdffitwindow=true,         
  pdftitle={DiSquawk: 512 cores, 512 memories, 1 JVM},  
  pdfborder={ 0 0 0 },        
  pdfauthor = {}, 
  pdfcreator = {}, 
  }

\usepackage[ruled]{algorithm2e}

\SetAlFnt{\small}
\SetAlCapFnt{\small}
\SetAlCapNameFnt{\small}
\SetAlCapHSkip{0pt}

\usepackage{amssymb}
\usepackage{mathpartir}

\makeatletter
\newcommand{\@chapapp}{\relax}%
\makeatother

\PassOptionsToPackage{hyphens}{url}

\usepackage[capitalise]{cleveref}
\crefformat{section}{\S#2#1#3}

\usepackage{xspace}
\usepackage{booktabs}
\usepackage{array}
\usepackage[utf8]{inputenc}

\usepackage{graphicx}

\graphicspath{{./figs/}}
\usepackage{fancyvrb}

\usepackage[usenames,dvipsnames,svgnames,table]{xcolor}

\usepackage{listings}
\usepackage{enumerate}
\usepackage[shadow]{todonotes}
\presetkeys%
    {todonotes}%
    {inline,backgroundcolor=yellow}{}
\usepackage{tikz}
\usepackage{amssymb}
\usepackage{amsthm}
\usepackage[shortlabels,inline]{enumitem}
\usepackage{mathtools}

\usetikzlibrary{shapes,decorations.pathmorphing,calc,shadows,shapes.multipart}

\newcommand{\order}[1]{\leq_{#1}}

\newcommand{\sa}{\mathrm{SA}}

\newcommand{\tuple}[1]{\ensuremath{\langle #1 \rangle}}
\newcommand{\hb}{\order{hb}}
\newcommand{\hbD}{\order{hb}^d}
\newcommand{\so}{\order{so}}
\newcommand{\soD}{\order{so}^d}
\newcommand{\po}{\order{po}}
\newcommand{\poD}{\order{po}^d}

\newcommand{\swD}{\order{sw}^d}

\newcommand{\aset}[1]{\{ #1 \}}
\newcommand{\paren}[1]{\left( #1 \right)}
\newcommand{\dom}[1]{\mathit{dom}\paren{#1}}
\newcommand{\rng}[1]{\mathit{rng}\paren{#1}}

\newcommand{\vol}[1]{\mathit{volatile}\paren{#1}}

\newcommand{\sfmt}[1]{\ensuremath{\mathsf{#1}}}
\newcommand{\sunit}{\sfmt{()}}
\newcommand{\sassign}[2]{#1 \coloneqq #2}
\newcommand{\snew}[1]{\sfmt{new} \; C(#1)}
\newcommand{\sheap}{\mathcal{H}}
\newcommand{\sscache}{\mathcal{C}}
\newcommand{\ssdcache}{\mathcal{D}}
\newcommand{\scaches}{\vec{\sscache}}
\newcommand{\ssdcaches}{\vec{\ssdcache}}
\newcommand{\scache}[1]{\sscache_{#1}}
\newcommand{\sdcache}[1]{\ssdcache_{#1}}
\newcommand{\sproc}[2]{#1\langle #2 \rangle}
\newcommand{\slet}[4]{\sfmt{let} \; #1 : #2 = #3 \; \sfmt{in}\; #4}
\newcommand{\sifelse}[3]{\sfmt{if~} #1 \sfmt{~then~} #2 \sfmt{~else~} #3}
\newcommand{\strue}{\sfmt{true}}
\newcommand{\sfalse}{\sfmt{false}}
\newcommand{\stids}{\textsc{Cid}\mathrm{s}}

\newcommand{\rulename}[1]{\textsc{#1}}
\newcommand{\inference}[3][]{\inferrule[{[#1]}]{#2}{#3}}
\newcommand{\inferencel}[3][]{\inferrule*[width=\linewidth,left={[#1]}]{#2}{#3}}

\newcommand{\rfield}{\rulename{Field}}
\newcommand{\rfieldd}{\rulename{FieldDirty}}
\newcommand{\rassign}{\rulename{Assign}}
\newcommand{\rvolatileread}{\rulename{VolatileRead}}
\newcommand{\rvolatilereadl}{\rulename{VolatileReadL}}
\newcommand{\rvolatilewrite}{\rulename{VolatileWrite}}
\newcommand{\rvolatilewritel}{\rulename{VolatileWriteL}}
\newcommand{\rfetch}{\rulename{Fetch}}
\newcommand{\rwriteback}{\rulename{WriteBack}}
\newcommand{\rinvalidate}{\rulename{Invalidate}}
\newcommand{\rmonitorentera}{\rulename{MonitorEnter}}
\newcommand{\rmonitorenterb}{\rulename{NestedMonitorEnter}}
\newcommand{\rmonitorexita}{\rulename{MonitorExit}}
\newcommand{\rmonitorexitb}{\rulename{NestedMonitorExit}}
\newcommand{\racquire}{\rulename{Acquire}}
\newcommand{\rrelease}{\rulename{Release}}
\newcommand{\rmigrate}{\rulename{Migrate}}
\newcommand{\rnew}{\rulename{New}}
\newcommand{\rspawn}{\rulename{Spawn}}
\newcommand{\rstart}{\rulename{Start}}
\newcommand{\rfinish}{\rulename{Finish}}
\newcommand{\rjoin}{\rulename{Join}}
\newcommand{\rinterrupt}{\rulename{Interrupt}}
\newcommand{\rinterruptedt}{\rulename{InterruptedT}}
\newcommand{\rinterruptedf}{\rulename{InterruptedF}}

\newcommand{\rcondt}{\rulename{IfTrue}}
\newcommand{\rcondf}{\rulename{IfFalse}}

\newcommand{\rcall}{\rulename{Call}}
\newcommand{\rred}{\rulename{Let}}
\newcommand{\rctxstep}{\rulename{CtxStep}}
\newcommand{\rpar}{\rulename{Blocked}}
\newcommand{\rlift}{\rulename{Lift}}

\newcommand{\rparg}{\rulename{ParG}}

\newcommand{\tunit}{\mathit{Unit}}

\newcommand{\tbool}{\mathit{Bool}}
\newcommand{\tnat}{\mathit{Nat}}

\newcommand{\lang}{DJC\xspace}

\newcommand{\wf}[1]{\textbf{WF-#1}}
\newcommand{\wfe}[1]{\textbf{WFE-#1}}
\newcommand{\wfh}[1]{\textbf{WFH-#1}}

\newcommand{\toa}[1]{\xrightarrow{\mathit{#1}}}

\newcommand{\init}{
  \ensuremath{ \{(r_t \mapsto \sfmt{VMThread}(\emptyset, \sfmt{spawned}))\};
               \emptyset;
               \emptyset
               \vdash \sproc{c}{r_t, \sfmt{start}}}}

\newcommand{\sclass}{
        \sfmt{class} \; C(\overrightarrow{f:\tau})\{e\}\{\vec{M}\}}

\newcommand{\smethod}{
        m(\overrightarrow{x:\tau})\{\sfmt{return}\;e;\} : \tau}
\crefname{figure}{Figure}{Figures}
\Crefname{figure}{Figure}{Figures}
\crefname{axiom}{axiom}{axioms}
\Crefname{axiom}{Axiom}{Axioms}

\newtheorem{lemma}{Lemma}

\newtheorem{theorem}{Theorem}

\newtheorem{case}{Case}[lemma]

\newtheorem{tcase}{Case}[theorem]

\begin{document}






\title{
  DiSquawk: 512 cores, 512 memories, 1 JVM
}
%
%
%
%
%


\author{
%
%
  Foivos S.~Zakkak\\
  {FORTH-ICS and University of Crete}\\
  \texttt{zakkak@ics.forth.gr}
  \and
  Polyvios Pratikakis\\
  {FORTH-ICS}\\
  \texttt{polyvios@ics.forth.gr}
}
\date{}

\maketitle

\begin{center}
  \textbf{\large Technical Report FORTH-ICS/TR-470, June 2016}
\end{center}

\begin{abstract}
  Trying to cope with the constantly growing number of cores per
  processor, hardware architects are experimenting with modular
  non cache coherent architectures.  Such architectures delegate the
  memory coherency to the software.  On the contrary, high productivity
  languages, like Java, are designed to abstract away the hardware
  details and allow developers to focus on the implementation of their
  algorithm.  Such programming languages rely on a process virtual
  machine to perform the necessary operations to implement the
  corresponding memory model.  Arguing, however, about the correctness of such
  implementations is not trivial.

  In this work we present our implementation of the Java Memory Model in
  a Java Virtual Machine targeting a 512-core non cache coherent memory
  architecture.
  We shortly discuss design decisions and present early evaluation
  results, which demonstrate that our implementation scales with the
  number of cores up to 512 cores.
  We model our implementation as the operational semantics of a Java
  Core Calculus that we extend with synchronization actions, and prove
  its adherence to the Java Memory Model.
\end{abstract}

%
%



%
%

%
%



\section{Introduction}
\label{sec:intro}

Current multicore processors rely on hardware cache coherence to
implement shared memory abstractions.  However, recent literature
largely agrees that existing coherence implementations do not scale
well with the number of processor cores, incur large energy and area
costs, increase on-chip traffic, or limit the number of cores per
chip~\cite{Choi2011,yang2014coherence,carter2013runnemede},
despite several attempts to design less costly or more scalable
coherence protocols~\cite{martin2012coherence,menezo2013coherence}.

To address that issue, recent work on hardware design proposes modular
many-core architectures.
Such examples are the Intel\textsuperscript{\textregistered}
Runnemede~\cite{carter2013runnemede} architecture, the Formic
prototype~\cite{lyberis2012formic}, and the EUROSERVER
architecture~\cite{euroserver}.
These architectures are designed in a way that allows scaling up by
plugging in more modules.
Each module is self-contained and able to interface with other modules.
Connecting multiple such modules builds a larger system that can be seen
as a single many-core processor.
In such architectures the trend is to use multiple mid-range cores with
local scratchpads interconnected using efficient communication channels.

The lack of cache coherence renders the software responsible for
performing the necessary data transfers to ensure data coherency in
parallel programs.
However, in high productivity languages, such as Java, the memory
hierarchy is abstracted away by the process virtual machines rendering
the latter responsible for the data transfers.
Process virtual machines
provide the same language guarantees to the developers as in
cache coherent shared-memory architectures.
Those guarantees are formally defined in the language's memory model.
The efficient implementation of a language's memory model on
non cache coherent architectures is not trivial though.
Furthermore, arguing about the implementation's correctness is even more
difficult.



In this work we present an implementation of the Java Memory Model
(JMM)~\cite{Manson:2005:JMMconf} in DiSquawk, a Java Virtual Machine
targeting the Formic-cube, a 512-core non cache coherent prototype based
on the Formic architecture~\cite{lyberis2012formic, formic}.
We shortly discuss design decisions and present evaluation results,
which demonstrate that our implementation scales with the number of
cores.
To prove our implementation's adherence to the Java Memory Model, we
model it as the operational semantics of Distributed Java Calculus
(\lang), a Java Core Calculus that we define for that purpose.




Specifically, this work makes the following contributions:
\begin{itemize}
\item We present a Java Memory Model (JMM) implementation for
  non cache coherent architectures that scales up to 512 cores, and we
  shortly discuss our design decisions.
\item We present Distributed Java Calculus (\lang), a Java core calculus
  with support for Java synchronization actions and explicit cache
  operations.
\item We model our JMM implementation as the operational semantics of
  \lang.
\item We prove that the operational semantics of \lang adheres to JMM
  and present the proof sketch.
\end{itemize}

The remainder of this paper is organized as follows.
\Cref{sec:jdmm} shortly presents JDMM, a JMM extension for
non cache coherent memory architectures, and the motivation for this
work;
\Cref{sec:implementation} presents our implementation of JDMM and
shortly discusses the design decisions;
\Cref{sec:fj} presents \lang, its operational semantics, and a proof
sketch of its adherence to JDMM;
%
%
\Cref{sec:related} discusses related work; and
\Cref{sec:conclusions} concludes.

\section{Background and Motivation}
\label{sec:jdmm}

In order to reduce network traffic and execution time, Java Virtual
Machines (JVMs) on non cache coherent architectures usually implement
some kind of software caching~\cite{McIlroy2010,Antoniu:2001} or
software distributed shared
memory~\cite{Yu1997,Veldema99,Zhu:2002,Factor:2003}.
Both approaches rely on similar operations; to access a remote object
they \emph{fetch} a local copy; to make dirty copies globally visible
they write them back (\emph{write-back}); and to free space in the cache
or force an update on the next access they \emph{invalidate} local
copies.
Since JMM~\cite{Manson:2005:JMMconf} is agnostic about such operations,
we base our work on the Java Distributed Memory Model
(JDMM)~\cite{zakkak:jdmm}.

The JDMM is a redefinition of JMM for distributed or
non cache coherent memory architectures.
It extends the JMM with cache related operations and formally defines when
such operations need to be executed to preserve JMM's properties.
The JDMM is designed to be as relaxed as the JMM\@.
Following a similar approach to that of Owens \em et al.\em~\cite{owens2009better} in the
x86 Total Store Order (x86-TSO) definition, the JDMM first defines an
abstract machine model and then defines the memory model based on it.




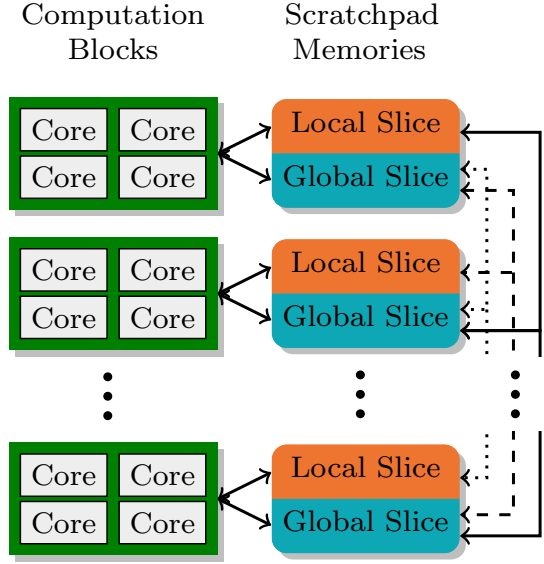
\begin{figure}[t]
\centering
\resizebox{0.6\textwidth}{!}{
  \definecolor{hwt}{HTML}{EEEEEE}
  \definecolor{cachemem}{HTML}{ED752F}
  \definecolor{memback}{HTML}{DDDDDD}
  \definecolor{globalback}{HTML}{0EA7B5}
  \definecolor{hwtback}{HTML}{FFBE4F}
  \tikzstyle{readwrite} = [<->, thick]
  \tikzstyle{gmemory} = [draw, fill=memback, dashed, rounded
  corners, minimum width=3.25cm, minimum height=6.5cm]
  \tikzstyle{lmemory} = [rounded corners,
  rectangle split, rectangle split parts=2, rectangle split part
  fill={cachemem,globalback}, rectangle split draw splits=false,
  minimum width=1.4cm]
  \def\mystrut{\vrule height .2cm depth .1cm width 0pt}
  \tikzstyle{hwthreads} = [fill=green!50!black, align=center, drop
  shadow]
  \tikzstyle{hwthread} = [draw, fill=hwt, align=center, minimum height=.4cm]
  \begin{tikzpicture}[node distance=.25cm]
    \scriptsize
    \node (Blocks) [text width=2cm, align=center] {Computation Blocks};
    \node (HTs1) [hwthreads, below=of Blocks] {%
      \tikz\node [hwthread] {Core};
      \tikz\node [hwthread] {Core}; \\
      \tikz\node [hwthread] {Core};
      \tikz\node [hwthread] {Core};
    };
    \node (HTs2) [hwthreads, below=of HTs1] {%
      \tikz\node [hwthread] {Core};
      \tikz\node [hwthread] {Core}; \\
      \tikz\node [hwthread] {Core};
      \tikz\node [hwthread] {Core};
    };
    \node (dots1) [below= 1mm of HTs2, align=center] {%
      \rotatebox{90}{\tikz\filldraw (0,0) circle (1pt);
        \tikz\filldraw (0,0) circle (1pt);
        \tikz\filldraw (0,0) circle (1pt);}
    };
    \node (HTs3) [hwthreads, below= 1mm of dots1] {%
      \tikz\node [hwthread] {Core};
      \tikz\node [hwthread] {Core}; \\
     \tikz\node [hwthread] {Core};
      \tikz\node [hwthread] {Core};
    };
    \scriptsize
    \node (Memory) [text width=1.4cm, , right=.4cm of Blocks, align=center] {Scratchpad Memories};
    \node [lmemory, right=.5cm of HTs1, drop shadow]
    {\nodepart{one}\mystrut Local Slice\nodepart{two}\mystrut Global Slice}; 
    \node (LM1) [lmemory, right=.5cm of HTs1] {\nodepart{one}\mystrut Local Slice\nodepart{two}\mystrut Global Slice};
    \node [lmemory, right=.5cm of HTs2, drop shadow] {\nodepart{one}\mystrut Local Slice\nodepart{two}\mystrut Global Slice}; 
    \node (LM2) [lmemory, right=.5cm of HTs2] {\nodepart{one}\mystrut Local Slice\nodepart{two}\mystrut Global Slice};
    \node (dots2) [below= 1mm of LM2, , align=center] {%
      \rotatebox{90}{\tikz\filldraw (0,0) circle (1pt);
        \tikz\filldraw (0,0) circle (1pt);
        \tikz\filldraw (0,0) circle (1pt);}
    };
    \node [lmemory, right=.5cm of HTs3, drop shadow] {\nodepart{one}\mystrut
      Local Slice\nodepart{two}\mystrut Global Slice}; 
    \node (LM3) [lmemory, right=.5cm of HTs3] {\nodepart{one}\mystrut
      Local Slice\nodepart{two}\mystrut Global Slice};
    \draw [readwrite] (HTs1.east) -- ($(LM1.west) + (0,.25cm)$);
    \draw [readwrite] (HTs1.east) -- ($(LM1.west) - (0,.25cm)$);
    \draw [readwrite] (HTs2.east) -- ($(LM2.west) + (0,.25cm)$);
    \draw [readwrite] (HTs2.east) -- ($(LM2.west) - (0,.25cm)$);
    \draw [readwrite] (HTs3.east) -- ($(LM3.west) + (0,.25cm)$);
    \draw [readwrite] (HTs3.east) -- ($(LM3.west) - (0,.25cm)$);
    \draw [readwrite]
    ($(LM1.east) + (0,.2cm)$) --
    ($(LM1.east) + (.75cm,.2cm)$) --
    ($(LM2.east) + (.75cm,-.35cm)$) --
    ($(LM2.east) + (0,-.35cm)$);
    \draw [readwrite]
    ($(LM2.east) + (0,-.35cm)$) --
    ($(LM2.east) + (.75cm,-.35cm)$) --
    ($(LM3.east) + (.75cm,-.35cm)$) --
    ($(LM3.east) + (0,-.35cm)$);
    \draw [readwrite, dashed]
    ($(LM1.east) + (0,-.35cm)$) --
    ($(LM1.east) + (.5cm,-.35cm)$) --
    ($(LM2.east) + (.5cm,.2cm)$) --
    ($(LM2.east) + (0,.2cm)$);
    \draw [readwrite, dashed, ->]
    ($(LM2.east) + (.5cm,.2cm)$) --
    ($(LM3.east) + (.5cm,-.15cm)$) --
    ($(LM3.east) + (0,-.15cm)$);
    \draw [readwrite, dotted]
    ($(LM1.east) + (0,-.15cm)$) --
    ($(LM1.east) + (.25cm,-.15cm)$) --
    ($(LM2.east) + (.25cm,-.15cm)$) --
    ($(LM2.east) + (0,-.15cm)$);
    \draw [readwrite, dotted, ->]
    ($(LM2.east) + (.25cm,-.15cm)$) --
    ($(LM3.east) + (.25cm,.2cm)$) --
    ($(LM3.east) + (0,.2cm)$);
    \node (dots3) [right=.9cm of dots2, fill=white, align=center,
    minimum height=.7cm, minimum width=.7cm] {%
      \rotatebox{90}{\tikz\filldraw (0,0) circle (1pt);
        \tikz\filldraw (0,0) circle (1pt);
        \tikz\filldraw (0,0) circle (1pt);}
    };
  \end{tikzpicture}%
  }
  \caption{The memory abstraction.}
  \label{fig:abstract-machine}

\end{figure}

\cref{fig:abstract-machine} presents an instance of the abstract
machine as presented in the JDMM paper.
On the left side there are several computation blocks with four
cores in each of them.  Each computation block connects
directly to its local scratchpad memory.  The scratchpad memory is
split in a local and a global slice.  In this model, each local slice
connects with every other global slice in the system, but not with any
local slice.  The connections are bi-directional: a core can copy data
from a remote global slice to the local cache to improve performance;
after finishing the job it can transfer back the new data.

The local slice of the scratchpad is used for the local data (\em i.e.\em,
Java stacks) and for caching remote data.  The global slices are
partitions of a total \em virtual \em Java Heap, similarly to Partitioned
Global Address Space (PGAS) models.  The state of the memory
can only be altered by the computation blocks or by committing a fetch,
a write-back, or an invalidate instruction.


In this abstract machine memory model the software needs to explicitly
transfer data in such a way that JMM guaranties are preserved.  At a
high level, JMM guarantees that data-race-free (DRF) programs are
sequentially consistent, and that variables cannot get
\emph{out-of-thin-air} values under any circumstances.  To
define our core calculus and couple it with the JDMM, we use a subset of
the notation used in the JDMM paper, which we present here along with
the JDMM short presentation.  The JDMM describes program executions as
tuples consisting of:
\begin{enumerate}[1)]
\setlength{\itemsep}{0em}
\item a set of instructions,
\item a set of actions, some of which are characterized as
  synchronization actions.

  The JDMM uses the following abbreviations to describe all possible kinds
  of actions:
  \begin{itemize}
  \item $R$ for read, $W$ for write, and $\mathit{In}$ for initialization of a
    heap-based variable,
  \item $\mathit{Vr}$ for read and $\mathit{Vw}$ for write of a volatile variable,
  \item $L$ for the lock and $U$ for the unlock of a monitor,
  \item $S$ for the start and $\mathit{Fi}$ for the end of a thread,
  \item $\mathit{Ir}$ for the interruption of a thread and
    $\mathit{Ird}$ for detecting such an interruption by another thread,
  \item $\mathit{Sp}$ for spawning (\verb!Thread.start()!) and $J$ for joining a
    thread or detecting that it terminated,
  \item $E$ for external actions, \em i.e.\em, I/O operations,
  \item $F$ for fetch from heap-based variables,
  \item $B$ for write-backs of heap-based variables,
  \item $I$ for invalidations of cached variables.
  \end{itemize}

  Note that actions with kind $\mathit{In}$, $\mathit{Ir}$,
  $\mathit{Ird}$, $\mathit{Vr}$, $\mathit{Vw}$, $L$, $U$, $S$,
  $\mathit{Fi}$, $\mathit{Sp}$, or $J$ are characterized as
  \emph{synchronization actions} and form the only communication
  mechanism between threads.

\item the program order, which defines the order of actions within
  each thread,
\item the synchronization order, which defines a total ordering among
  the synchronization actions,
\item the synchronizes-with order, which defines the pairs of
  synchronization actions ---release and acquire pairs, 
%
\item the happens-before order that defines a partial order among all
  actions and is the transitive closure of the program order and the
  synchronizes-with order, and
\item some helper functions that we do not use in this paper.
\end{enumerate}

The JDMM explicitly defines the conditions that a Java program
execution needs to satisfy on a non cache coherent architecture,
to be a well-formed execution.  These conditions are introduced
in~\cite[\S3 and \S4.2]{zakkak:jdmm}; we briefly present them
here.
Note that \wf{1}--\wf{9} were first introduced
in~\cite{Manson:2005:JMMconf}.


\begin{description}
  \setlength{\itemsep}{0.5ex}

\item[WF-1] \phantomsection \label{wf1}
  Each read of a variable sees a write to it.

\item[WF-2] \phantomsection \label{wf2}
  All reads and writes of volatile variables are volatile actions.

\item[WF-3] \phantomsection \label{wf3}
  The number of synchronization actions preceding another
  synchronization action is finite.

\item[WF-4] \phantomsection \label{wf4}
  Synchronization order is consistent with program order.

\item[WF-5] \phantomsection \label{wf5}
  Lock operations are consistent with mutual exclusion.

\item[WF-6] \phantomsection \label{wf6}
  The execution obeys intra-thread consistency.

\item[WF-7] \phantomsection \label{wf7}
  The execution obeys synchronization order consistency.

\item[WF-8] \phantomsection \label{wf8}
  The execution obeys happens-before consistency.

\item[WF-9] \phantomsection \label{wf9}
  Every thread's start action happens-before its other actions
  except for initialization actions.

\item[WF-10] \phantomsection \label{wf10}
  Every read is preceded by a write or fetch action,
  acting on the same variable as the read.

\item[WF-11] \phantomsection \label{wf11}
  There is no invalidation, update, or overwrite of a variable's
  cached value between the action that cached it and the read that
  sees it.

\item[WF-12] \phantomsection \label{wf12}
  Fetch actions are preceded by at least one write-back of
  the corresponding variable.

\item[WF-13] \phantomsection \label{wf13}
  Write-back actions are preceded by at least one write to
  the corresponding variable.

\item[WF-14] \phantomsection \label{wf14}
  There are no other writes to the same variable between a
  write and its write-back.

\item[WF-15] \phantomsection \label{wf15}
  Only cached variables can be invalidated.  Invalid cached
  data cannot be invalidated.

\item[WF-16] \phantomsection \label{wf16}
  Reads that see writes performed by other threads are
  preceded by a fetch action that fetches the write-back of the
  corresponding write and there is no other write-back of the
  corresponding variable happening between the write-back and the
  fetch.

\item[WF-17] \phantomsection \label{wf17}
  Volatile writes are immediately written back.

\item[WF-18] \phantomsection \label{wf18}
  A fetch of the corresponding variable happens
  immediately before each volatile read.

\item[WF-19] \phantomsection \label{wf19}
  Initializations are immediately written-back; their
  write-backs complete before the start of any thread.

\item[WF-20] \phantomsection \label{wf20}
  The happens-before order between two writes is consistent with
  the happens-before order of their write-backs.  
\end{description}

Two additional conditions must hold for executions
containing thread migration actions.  Intuitively:
\begin{description}
\item[WFE-1] \phantomsection \label{wfe1}
  There is a corresponding fetch action between a thread
  migration and every read action.

\item[WFE-2] \phantomsection \label{wfe2} Additionally, to make sure
  the fetched value is the latest according to the happens-before
  order, any dirty data on the \emph{old} core need to be
  written-back.
\end{description}

\begin{figure}[t]
  \centering
  \tikzstyle{thread} = [->, decorate,
  decoration={snake,amplitude=1mm,segment length=1cm},line width=1mm]
  \tikzstyle{action} = [text centered, fill=white]
  \tikzstyle{window} = [draw, circle, drop shadow,
  text centered, text width=2em]
  \begin{tikzpicture}[node distance=1.5cm]
    \node (T1) [color=BrickRed] {\textit{T1}};
    \node (T2) [color=ForestGreen, right of=T1] {\textit{T2}};
    \draw [thread, color=BrickRed] (T1.south) -- ++ (0,-4.5);
    \draw [thread, color=ForestGreen] (T2.south) -- ++ (0,-4.5);
    \node (l1) [color=BrickRed, action, below= 2mm of T1] {\texttt{m-enter}};
    \node (w1) [color=BrickRed, action, below= 2mm of l1] {\texttt{write}};
    \node (u1) [color=BrickRed, action, below= 4mm of w1] {\texttt{m-exit}};
    \node (l2) [color=ForestGreen, action, below= 2.2cm of T2] {\texttt{m-enter}};
    \node (r2) [color=ForestGreen, action, below= 4mm of l2] {\texttt{read}};
    \node (u2) [color=ForestGreen, action, below= 2mm of r2] {\texttt{m-exit}};
    \draw [dashed, thick, rounded corners] ($(w1.south west)-(5pt,0)$) rectangle
    ($(r2.north east)+(10pt,0)$);
    \draw [color=BrickRed, dashed, very thick, rounded corners]
    ($(w1.south west)-(5pt,0)$) rectangle (u1.north east);
    \draw [color=ForestGreen, dashed, very thick, rounded corners]
    (l2.south west) rectangle ($(r2.north east)+(10pt,0)$);
  \end{tikzpicture}
  \caption{Time window example.}
  \label{fig:time-window}
\end{figure}
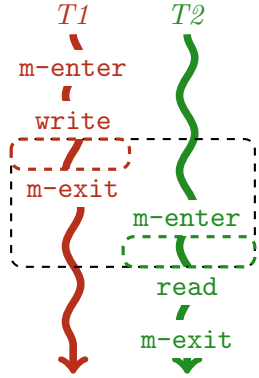

Note that, in the core JDMM, context switching without thread migration
is examined only as an extension.
As a result, we hereto use a slightly modified version of \wf{16}
to allow \lang to be more relaxed in the case of context switches and
still comply with the JDMM\@.
The modified rule enables different threads running on the same core
to share the contents of a single cache, without breaking the adherence
to JMM, as shown in~\cite[\S5.2]{zakkak:jdmm}.
That is:
\begin{description}
\item[WF-16] Reads that see writes performed by another \emph{core}
  are preceded by a fetch action that fetches the write-back of the
  corresponding write and there is no other write-back of the
  corresponding variable happening between the write-back and the
  fetch.
\end{description}


The JDMM intuitively states that a write-back and its corresponding fetch
may be executed any time in the time window between a write and the
corresponding read, given that the write happens-before\footnote{as
  defined in~\cite{lamport1978time}} this read.
For instance, in~\cref{fig:time-window} the thread \textit{T1} performs
a \texttt{write} that happens-before the corresponding \texttt{read} in
thread \textit{T2}.
The happens-before relationship is a result of the monitor release,
\texttt{m-exit}, by \textit{T1} and the subsequent monitor acquisition,
\texttt{m-enter}, by \textit{T2}.
The time window that the JDMM allows a write-back and its corresponding
fetch to be performed is marked with the big black dashed rectangle.

This flexibility on when these operations can be executed, allows for
great optimization in theory.  However, in practice it is very
difficult to even estimate this time window.  The JVM needs to keep
extra information for every field in the program and constantly update
it.  It needs to know the sequence of lock acquisition, who was the
last writer, if their write has been written-back, and whether the
cached value (if any) is consistent with the main memory or not.
Implementing these over software caching seems prohibitive, as the
cost of the bookkeeping and the extra communication is expected to
be much higher than the expected benefits regarding energy, space, and
performance.

\begin{figure}[bt]
  \centering
  \includegraphics{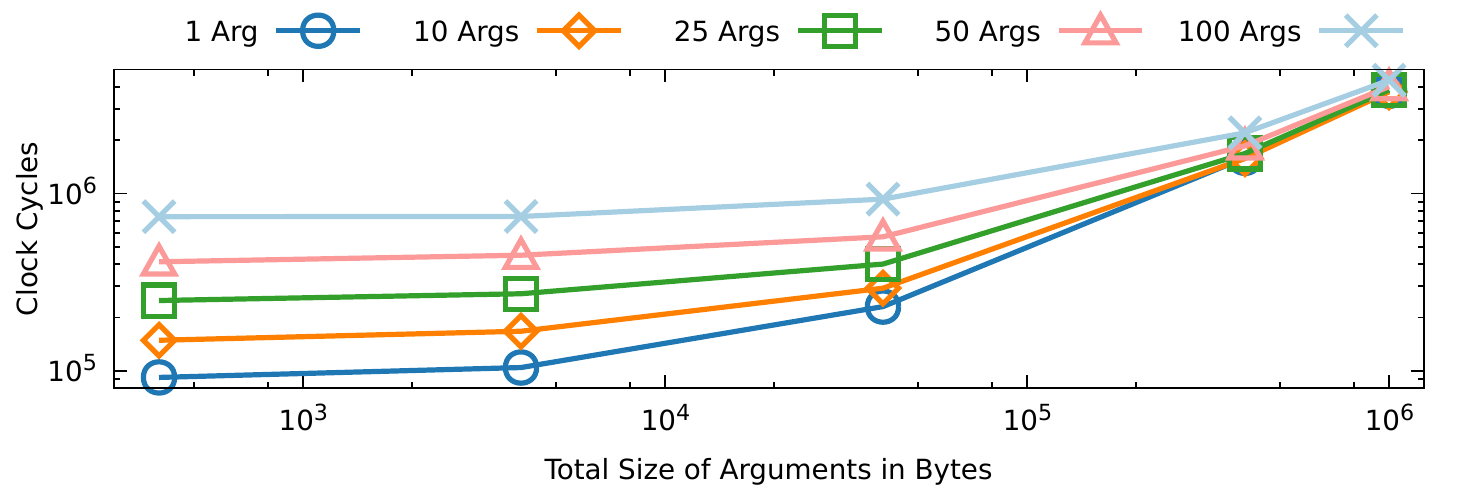}
  \caption{Performance impact of arguments size}
  \label{fig:arg-size-num}
\end{figure}

An intuitive implementation is to issue all the write-backs at release
actions.
However, this may result in long blocking release actions for critical
sections that perform writes on large memory segments.
To demonstrate the overhead of such operations we perform a simple
experiment, where one core transfers a given data set from another
core's scratchpad to its own.
\cref{fig:arg-size-num} shows the impact of the arguments' size and
number on the data transfer time.
On the y-axes we plot the clock cycles consumed to transfer all the data
from one core's to another core's scratchpad.
On the x-axes we plot the total size of the data in Bytes.
Each line in the plot represents a different partitioning of the data,
in 1, 10, 25, 50, and 100 arguments respectively.
We observe that apart from the total data size the partitioning of the
data impacts the transfer time as well.
This is a result of performing multiple data transfers instead of a
single bulk transfer.
As a result, keeping a lot of dirty data cached until a release
operation is expected to perform badly, as it most probably will need to
perform multiple data transfers to write-back non contiguous dirty data.

Hera-JVM~\cite{McIlroy2010} ---the only, to the best of our knowledge,
JVM for a non cache coherent architecture that claims adherence
to the JMM--- issues a write-back for every write and then waits for all
pending write-backs to complete at release actions.
This approach significantly reduces the blocking time at release
actions, but results in multiple redundant write-backs in cases where a
variable is written multiple times in a critical section.
Such redundant memory operations are usually overlapped with
computation, keeping their performance overhead low.
However, the additional energy consumption they impose might still be
significant in energy-critical systems.
Additionally, in the case of writing to array elements, their approach
results in one memory transfer per element when a bulk transfer can be
used to improve performance and energy efficiency.

In this work we propose an alternative policy regarding write backs,
that aims to mitigate such cases by caching dirty data up to a certain
threshold.
Additionally, since the Formic architecture is more relaxed than the
Cell B.E.~\cite{cellBE} architecture that Hera-JVM is targeting, we
also present novel mechanisms to handle synchronization.

\section{Implementation}
\label{sec:implementation}

We implement our memory and cache management policy in DiSquawk, a
JVM we developed for the Formic-cube 512-core prototype.
Formic-cube is based on the Formic
architecture~\cite{lyberis2012formic}, which is modular and allows
building larger systems by connecting multiple smaller modules.
The basic module in the Formic architecture is the Formic-board.
Each board consists of 8 MicroBlaze\texttrademark-based,
non cache coherent cores and is equipped with 128MB of scratchpad
memory.
Each core also features a private software-managed, non-coherent,
two-level cache hierarchy; a hardware queue (mailbox) that supports
concurrent en-queuing, and de-queuing only by the owner core; and a DMA
engine.
All of Formic's scratchpads are addressable using a global address
space, and data are transferred through DMA transfers and mailbox
messages to and from remote memory addresses.

\subsection{Software Cache Management}
\label{sec:implementation:cache}

As the Formic-cube does not provide hardware cache coherence, we
build our JVM based on software caching.
Each core is assigned a part of the local scratchpad, which it uses as
its private software cache.
This software cache is entirely managed by the JVM, transparently to
the programmer.

To limit the amount of cached dirty data up to a given threshold we
split the software cache in two parts.
The first part, called \emph{object cache}, is used for caching objects
and is append-only ---writes on this cache are not permitted.
The second part, called \emph{write buffer}, is dedicated to
caching dirty data.
When the write buffer becomes full, we write back all its data and
update the corresponding fields in the object cache, if the
corresponding object is still cached.
Note that the combination of the write-buffer and the object cache form
a memory-hierarchy, where the write-buffer is below the object cache.
That is, read accesses first go through the write-buffer and only if
they miss they go to the object cache.
If they miss again, the JVM proceeds to fetch the corresponding
object.
This way, we
\begin{enumerate*}[label=\itshape\alph*\upshape)]
\item set an upper limit on the release operations' blocking time;
\item allow for overlapping write-backs with computation when the
  threshold is met;
\item allow for bulk transfer of contiguous data, \em e.g.\em, written elements
  of an array; and
\item allow for multiple writes to the same variable without the need
  to write back every time.
\end{enumerate*}
At acquisition operations, we write back all the dirty data, if any, and
invalidate both the object cache and the write buffer, in order to force
a re-fetch of the data if they get accessed in the future.
The write-back of the dirty data at acquisition operations is necessary
since we invalidate all the cached data.
Consider an example where a monitor is entered (acquire operation) then
a write is performed, and a different monitor is now entered (acquire
operation).
In this case simply invalidating all cached data, would result in the
loss of the write.

This approach is safe and sound, as we later show, but shrinks
the aforementioned time window thus limiting the optimization space.  A
visualization of the shrunk time window is presented
in~\cref{fig:time-window}.  The small red dashed rectangle on the
upper left corner of the big rectangle is the time window in which the
write-back can be executed.  Respectively the small green dashed
rectangle on the lower right corner is the time window in which the
corresponding fetch can be executed.
Note that although pre-fetching data, even in the shrunk time window,
allows for significant performance optimizations we do not implement it
in this work.
Alternatively, we only fetch data at cache misses.
Pre-fetching depends on program analysis to infer which data are going
to be accessed in the future.
Such analyses
are not specific to non cache coherent architectures or the Java Memory
Model, thus they our out of the scope of this work.

Despite the aforementioned reduction of flexibility regarding when a
data transfer can happen, and the lack of support for
pre-fetching, we are still able to achieve good performance and scale
with the number of cores due to the efficient on-chip communication
channels.
To demonstrate this we use the Crypt, SOR, and Series benchmarks from
the Java Grande~\cite{grande} suite and the Black-Scholes benchmark from
the PARSEC suite~\cite{Bienia:2008:PBS:1454115.1454128}, ported to Java.
Due to the lack of garbage collection and the upper limit of 4 GB heap
we are unable to run reasonable workloads with the rest of the Java
Grande benchmarks.
These benchmarks require larger than 4 GB datasets to produce meaningful
results on a large number of cores and some of them also create objects
with short lifespans, relying on garbage collection to reclaim their
memory.
Series and Black-Scholes are embarrassingly parallel benchmarks.
Each thread operates on a different subset of data from an input set and
creates a new set with the corresponding results.
The results are then accessed by the main thread for validation.
Crypt comprises of two embarrassingly phases.
In the first phase each thread encrypts a subset of the input data and
then waits on a barrier.
When all threads reach the barrier they proceed to decrypt each a subset
of the encrypted data.
The results are then compared to the original input for validation.
SOR performs a number of iterations where each thread acts on a
different block of an array accessing the previous and next neighboring
blocks as well.
As a result, each iteration depends on the neighboring blocks.
To ensure that the neighboring blocks are ready, SOR uses a volatile
counter for each thread.
This counter reflects the iteration the corresponding thread is on.
Each thread updates the counter at the end of each iteration and
accesses the two counters of the neighboring threads.



\begin{figure}[t]
  \centering
  \includegraphics{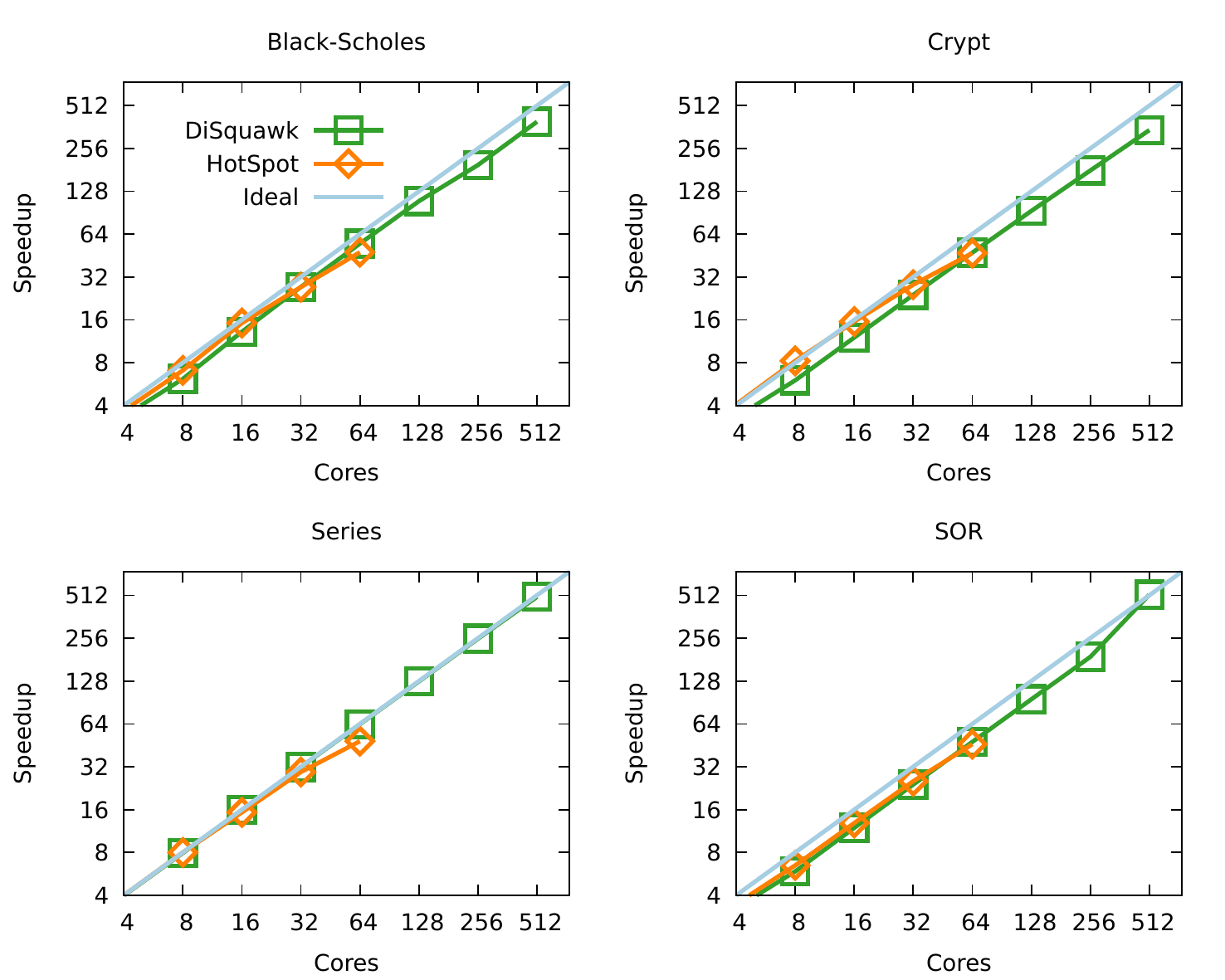}
  \caption{Speedup Results}
  \label{fig:speedup}
\end{figure}

\cref{fig:speedup} presents the speedup of the four benchmarks on both
DiSquawk, running on the formic-cube, and HotSpot running on a 4-chip
NUMA machine with 16 cores per chip, totalling 64 cores.
Since formic-cube is a prototype clocked at 10MHz, a comparison of the
throughput or the execution time is not possible, thus we chose to
compare the applications' scaling on both architectures.
The presented speedups are over the performance of the application
running on a single core on each architecture respectively.
Since DiSquawk does not support JIT compilation, we also disable it in
HotSpot (using the \texttt{-Xint} flag); this allows us to better
understand the applications' behavior on both architectures.
The number of Java threads, one per core, is placed on the x-axis, and
the speedup is placed on the y-axis.
Both axes are in logarithmic scale of base 2.
We observe that all benchmarks manage to scale with the number of cores
in both architectures.
Black-Scholes and Series scale better on DiSquawk than HotSpot when
using 32 or more cores, while Crypt performs better on HotSpot than
DiSquawk when using up to 32 cores.



\subsection{Java Monitors}
\label{sec:implementation:monitors}

Apart from the data movement, JDMM also dictates the operation of Java
monitors.  Java monitors are essentially re-entrant locks associated
with Java objects.  In Java, each object is implicitly associated with a
monitor and can be used in a \texttt{synchronized} block as the
synchronization point.  Java monitors are usually implemented using
atomic operations, such as \emph{compare and swap}, in shared-memory
cache coherent architectures, relying on the hardware to synchronize
multiple threads trying to obtain the monitor.  Such atomic operations
are not standard in non cache coherent
architectures, though~\cite{howard2010,lyberis2012formic}.

To implement the Java monitors on such architectures we propose a
synchronization manager: a server running on a dedicated core, handling monitor
enter/exit requests.  To keep contention at low levels we use multiple
synchronization managers according to the number of available cores on the
system.
Each synchronization manager is responsible for a number of objects in
the system, and each object can be associated with its synchronization
manager using a hash function.
When a thread executes a monitor-enter the JVM communicates with the
corresponding synchronization manager and requests ownership of the monitor.
This way all requests regarding a single monitor end up in the
corresponding synchronization manager's hardware message queue, from where they
are handled by the synchronization manager one by one, in the order they
arrived.  We essentially delegate the synchronization of the requests to
the architecture's network on chip, and provide mutual exclusion through
the synchronization managers.

To reduce the synchronization managers' load, the network's traffic and
contention, and to keep energy consumption low we take advantage of the
blocking nature of monitors.
Instead of sending back negative responses, when a monitor is already
acquired by some other thread, we queue the monitor-enter requests in
the synchronization manager, and assign the monitor to the oldest
requester when it becomes available.
This way we ensure fairness in the order that the requests are handled.
Although this is not required by the Java Language
Specification~\cite{JLS8}, we consider it better than arbitrarily
choosing one of the waiting threads, since it avoids the starvation of
threads.
Additionally, when a thread is waiting for a monitor it yields to free
up resources for other threads.
Instead of periodically rescheduling such waiting threads ---as we do
with other yielded threads--- we use a mechanism that reschedules them
only when the monitor they requested has been assigned to them.
That is, the synchronization manager has send an acknowledgement message
to the core executing the waiting thread.

Using a synthetic micro-benchmark which constantly issues requests to a
single monitor manager from $X$ cores in the system, where
$0 < X < 512$, we find that, on our system, at least one synchronization
manager per 243 cores is required to avoid scenarios where the
synchronization manager becomes a bottleneck.

\subsection{Volatile Variables}
\label{sec:implementation:volatiles}

Another challenging part is the support of volatile variables.  Volatile
variables are special, because accessing them is a form of
synchronization.  Specifically, volatile reads act as acquire
operations, while volatile writes act as release operations.  That said,
after a volatile read any data visible to the last writer of the
corresponding volatile variable must become visible to the reader.
Volatile accesses are usually implemented using memory fences provided
by the underlying architecture in shared-memory cache coherent
systems~\cite{lea2008jsr}.

Since non cache coherent architectures do not provide memory fences, in
our implementation we rely on synchronization managers to ensure a total
ordering between the various accesses to a volatile variable.
Essentially we treat volatile accesses as synchronized blocks protected
by a special monitor, unique per volatile variable.
Therefore, we write back and invalidate any cached data before
volatile accesses, and write back the dirty data immediately after
volatile writes.
This approach comes at the cost of unnecessary cache invalidations in
the case of volatile writes, which should not be often since volatile
variables are usually employed as a completion, interruption or status
flag~\cite[\S 3.1.4]{peierls2006java} ---meaning that they are being
mostly read during their life-cycle.

A side-effect of this implementation is the provision of mutual
exclusion to concurrent accesses on the same volatile variable.
Since Formic provides no guarantees about the atomicity of memory
accesses, we rely on this side-effect to ensure a volatile read will
never return an \emph{out-of-thin-air} value due to a partial update.


\subsection{Wait/Notify Mechanism}
\label{sec:implementation:waitnotify}
Java also offers the wait/notify mechanism, which allows a thread to
block its execution and wait for another thread to unblock it.
Since \texttt{wait()} and \texttt{notify()} require the monitor of the
corresponding object to be held by the executing thread, we use the
synchronization manager to keep track of such operations as well.  The
synchronization managers are holding a list of waiters for each object
they are responsible for.  Note that to keep the space overhead low we
only allocate records when the first request for an object arrives.
Initially, the synchronization managers hold no data for the objects
they are responsible for.  Whenever a thread invokes \texttt{wait()} a
special message is send to the synchronization manager that adds the
corresponding thread to the waiters queue and releases the monitor.  As
a result, before sending such messages we write back any dirty data.  To
support \texttt{wait()} invocations with a timeout we also support
messages to the synchronization manager that request the removal of a
thread from the waiters list.  When \texttt{notify()} is invoked it
sends a message to the synchronization manager, which notifies and
removes the longest waiting thread (if any).  In the case of
\texttt{notifyAll()}, all threads in the waiters queue get notified and
removed.

\subsection{Liveness Detection}
\label{sec:implementation:liveness}
For the detection of thread termination and checking of liveness we rely
on volatile variables.  Each thread is described using a JVM internal
object, which holds a volatile variable with the state of the thread.
The supported states are, \emph{spawned}, \emph{alive}, \emph{dead}.  We
implement \texttt{isAlive()} as a simple read to that state, if it
is equal to \emph{alive} then we return \texttt{true}.  On the other
hand, for the \texttt{join()} method we avoid spinning on the state
variable in an effort to reduce energy consumption and free up resources
for other threads in the system.  We base our \texttt{join()}
implementation on the \texttt{wait()}/\texttt{notify()} mechanism.
Since a thread invoking \texttt{join()} will have to wait until the
completion of the thread it joins, we yield it by invoking wait on the
JVM internal object, describing the thread.  When the corresponding
thread reaches completion it invokes \texttt{notifyAll()} on that
internal object and wakes up any joiners.

DiSquawk currently does not support interruptions.
We consider their implementation regarding synchronization to be straightforward.
Before sending an interrupt, all dirty data of the sending thread need
to be written back, and upon interruption the receiving thread needs to
write back any dirty data if present and invalidate its object cache.




\section{The Calculus}
\label{sec:fj}

To argue about the correctness of our implementation, we model it
using a Java core calculus and its operational semantics.
We base our calculus on the Java core calculus
introduced by Johnsen \em et al.\em~\cite{Johnsen2012257}, which omits
inheritance, subtyping,
and type casts, and adds concurrency and explicit lock support.
We extend that calculus by replacing the explicit lock
support with synchronization operations and adding support for cache
operations.
We define the operational semantics of the resulting Distributed Java
Calculus (\lang) and use it to argue about the correctness of the
cache and monitor management techniques used in DiSquawk.

\subsection{Syntax}
\label{sec:syntax}

\begin{figure}[bt]
  \vspace{-1ex}
  \begin{displaymath}
    \small
    \arraycolsep=4pt
    \begin{array}{lrcl}
      \textrm{Program}     & J      & \Coloneqq  & \vec{D}
      \\
      \textrm{Class Def.}  & D      & \Coloneqq  & \sclass \\
      \textrm{Types}       & \tau   & \Coloneqq  & C
                                                   \mid \tbool
                                                   \mid \tnat
                                                   \mid \tunit
      \\
      \textrm{Methods}     & M      & \Coloneqq  & \smethod \\
      \textrm{Expressions} & e      & \Coloneqq  & x
                                                   \mid \snew{\vec{e}}
                                                   \mid e.f
                                                   \mid \sassign{e.f}{e}
      \\
                           &        & \mid       & \slet{x}{\tau}{e}{e}
      \\
                           &        & \mid       & \sifelse{e}{e}{e}
                                                   \mid e.m(\vec{e})
      \\
                           &        & \mid       &
                                                   e.\sfmt{acquire}
                                                   \mid e.\sfmt{release}
      \\
                           &        & \mid       & e.\sfmt{monitorenter}
                                                   \mid e.\sfmt{monitorexit}
      \\
      \textrm{Values}      & v      & \Coloneqq  & r
                                                   \mid \sunit
                                                   \mid \strue \mid \sfalse \mid n
      \\
      \textrm{Contexts}    & \mathit{E}(\bullet) & \Coloneqq  &
                                                                \snew{v,\ldots,\bullet,\ldots,e}
                                                                \mid \bullet.f
      \\
                           &        & \mid       &
                                                   \sassign{e.f}{\bullet}
                                                   \mid \sassign{\bullet.f}{v}
      \\
                           &        & \mid       & \slet{x}{\tau}{\bullet}{e}
      \\
                           &        & \mid       & \sifelse{\bullet}{e}{e}
      \\
                           &        & \mid       & e.m(v,\ldots,\bullet,\ldots,e)
      \\
                           &        & \mid       & \bullet.\sfmt{monitorenter}
                                                   \mid \bullet.\sfmt{monitorexit}
      \\
      \textrm{Threads}     & T      & \Coloneqq  & \sproc{c}{r, \sfmt{start}}
                                                   \mid \sproc{c}{r, e}
                                                   \mid (T \parallel T)
                                                   \mid \mathbf{0}
      \\
      \textrm{Object}      & o      & \doteq     & C(\overrightarrow{f \mapsto v})
                                                   \mid C(\overrightarrow{f \mapsto v}, \sfmt{started}) \\
                           &        & \mid       & C(\overrightarrow{f \mapsto v}, \sfmt{spawned}) \\
                           &        & \mid       & C(\overrightarrow{f \mapsto v}, \sfmt{finished}) \\
                           &        & \mid       & C(\overrightarrow{f \mapsto v}, \sfmt{interrupted}) \\
      \textrm{Heap}       & \sheap & \doteq     & \overrightarrow{r \mapsto (o, l)} \\
      \textrm{Object Cache} & \sscache & \doteq  & \overrightarrow{r \mapsto o} \\
      \textrm{Write Buffer} & \ssdcache & \doteq & \overrightarrow{r.f \mapsto v} \\
      \textrm{Cache per Core} & \vec{\sscache} & \doteq& \overrightarrow{c \mapsto \sscache} \\
      \textrm{Buffer per Core} & \vec{\ssdcache} & \doteq& \overrightarrow{c \mapsto \ssdcache} \\
      \textrm{Lock State}  & l      & \Coloneqq  & 0 \mid r(n) \\
    \end{array}
  \end{displaymath}
  \caption{Abstract syntax of \lang}
  \label{tbl:syntax}
\end{figure}

The syntax of \lang is presented in~\cref{tbl:syntax}.  A Java
program $J$ consists of a sequence $\vec{D}$ of class definitions.
A class is defined as $\sclass$ where $C$ is the class name;
$\overrightarrow{f:\tau}$ is the list of field declarations, where each
$f_i$ is unique; $e$ is the body of the class constructor; and $\vec{M}$
is a sequence of method definitions.
The calculus types are class names $C$, boolean scalar types $\tbool$,
scalar natural numbers $\tnat$, and $\tunit$ for the unit value $()$.
A method is defined as $\smethod$ where $m$ is the method's name;
$\overrightarrow{x:\tau}$ is the set of formal arguments; $e$ is the
method body; and $\tau$ is the return type.
To keep the calculus simple we do not support method overloading.

The syntax includes variables $x$; creation of class instances as
$\snew{\vec{e}}$; field accesses as $e.f$, where $f$ is a unique field
identifier; field updates as $\sassign{r.f}{e}$; and sequential
composition using the let-construct as $\slet{x}{\tau}{e}{e}$.
Note that the evaluation of $e$ may have side-effects.
Conditional expressions are expressed as
$\sifelse{e}{e}{e}$; and method calls as $e.m(\vec{e})$, where $m$ is
the method name.

The syntax also includes monitor enter and exit actions
as expressions $e.\sfmt{monitorenter}$ and
$e.\sfmt{monitorexit}$, respectively.
Note that volatile accesses do not have separate bytecodes in Java; they
appear as normal memory accesses and the JVM checks at runtime whether
they are volatile or not. Thus, we do not provide special syntax for
them.

Values $v$ are references to objects $r$, the unit value $\sunit$,
boolean constants $\strue$ and $\sfalse$ and scalar numerical constants
$n$, abstracting over all other Java scalar types.
Contexts are used to show the evaluation sequence of the expressions.
In each expression in $E(\bullet)$ the $\bullet$ is evaluated first.

To argue about threads at runtime we extend \lang's syntax with run-time
threads.
A thread is defined as $\sproc{c}{r, \sfmt{start}}$ or
$\sproc{c}{r, e}$, where $c$ is the \emph{unique} identification of the
core that executes it; $r$ is the corresponding instance of the
\texttt{Thread} class; $\sfmt{start}$ is the thread start action, that
signals the start of its execution and is not to be confused with the
\texttt{start()} method of the \texttt{Thread} class; and $e$ is the
thread's body.
Threads can be composed in parallel pairs using the associative and
commutative binary operator $\parallel$.
The empty thread is marked with $\mathbf{0}$ and is the neutral element
of $\parallel$.

We represent an object in the runtime syntax as
$C(\overrightarrow{f \mapsto v})$ or
$C(\overrightarrow{f \mapsto v}, \mathit{state})$.
The first form is used for every object in the memory, while the second
is only used for thread objects whose \texttt{start()} method has
been invoked, and $\mathit{state}$ can be one of $\sfmt{spawned}$,
$\sfmt{started}$, $\sfmt{finished}$, and $\sfmt{interrupted}$.
Each object contains the name of its class and a map of field names $f$
to values $v$.
A thread whose \texttt{start()} method has been invoked is \em
spawned\em.
A thread whose \textrm{run()} method has been invoked is \em
started\em.
A thread that has reached completion is \em finished\em.
A thread whose \texttt{interrupt()} method has been
invoked is \em interrupted\em.

The memory of the system is split into the Heap $\sheap$, the object cache
$\sscache$, the write buffer $\ssdcache$, the object cache per core
$\scaches$, and the write buffer per core $\ssdcaches$.
The heap is a map from references $r$ to objects $o$ and their monitor
$l$.
The object cache is a map from references $r$ to objects $o$.
The write buffer is a map from object fields $r.f$ to values $v$.
The object cache per core is a map from core ids $c$ to object caches
$\sscache$.
Similarly, the write buffer per core is a map from core ids $c$ to write
buffers $\ssdcache$.

To model mutual exclusion we also add a lock state to the runtime
syntax.
A lock $l$ may be free, i.e., $0$, or acquired by some thread $r$, $n$
times.

\subsection{Operational Semantics}
\label{sec:oper-semant}

The operational semantics of \lang are based on those introduced by
Johnsen \em et al.\em~\cite{Johnsen2012257}.
In this work we introduce new rules for \emph{fetch}, \emph{write-back},
\emph{invalidate}, \emph{volatile-read}, \emph{volatile-write},
\emph{start}, \emph{finish}, \emph{join}, \emph{interrupt},
\emph{interrupt detection}, and \emph{migrate} operations.
Note that we do not model \texttt{java.util.concurrent}, a Java
library providing more synchronization mechanisms, in our
formalization, since its interference with JMM is not yet fully defined.

\begin{figure}[t]
  \centering
  \small
  \hfill
  \begin{tabular}{r p{0.7\linewidth}}
    \textbf{Notation}               & \textbf{Definition} \\
                                    & \\[-.5em]
    $r$                             & Reference value\\
    $m$                             & Method identifier\\
    $f$                             & Field identifier\\
    $c$                             & Core identifier\\
    $\dom{X}$                       & Returns the keys of the map $X$\\
    $\rng{X}$                       & Returns the values of the map $X$\\
    $\vec{X}[X_i'/X_i]$             & Replaces $X_i$ with $X_i'$ in $X$\\
    $\vec{X}\downarrow\vec{x}$      & The subset of map bindings in $X$ with keys in $\vec{x}$ \\
    $\vol{r.f}$                     & Returns true if $r.f$ is \texttt{volatile}\\
    $C\paren{\overrightarrow{f \mapsto v}}$ & A Java object that is an instance of class $C$ with mappings of field names to values $\overrightarrow{f \mapsto v}$ \\
  \end{tabular}
  \hfill~
\caption{Definition of Notation}
\label{tbl:defs}
\end{figure}

\begin{figure}[th!]
  \vspace{-1ex}
  \begin{displaymath}
    \small
    \begin{array}{c}
      \multicolumn{1}{l}{
      \fbox{$
      \sheap; \sscache; \ssdcache \vdash \sproc{c}{r_t, e}
      \toa{\alpha}
      \sheap; \sscache; \ssdcache \vdash \sproc{c}{t_r, e}
      $}
      }
      \\ \\

      \inferencel[\rctxstep]
      {
      \sheap; \sscache; \ssdcache \vdash \sproc{c}{r_t, e}
      \toa{\alpha}
      \sheap'; \sscache'; \ssdcache' \vdash \sproc{c}{r_t, e'}
      }
      {
      \sheap; \sscache; \ssdcache \vdash \sproc{c}{r_t, E(e)}
      \toa{\alpha}
      \sheap'; \sscache'; \ssdcache' \vdash \sproc{c}{r_t, E(e')}
      }
      \\ \\

      \inference[\rcondt]
      {}
      {
      \sheap; \sscache; \ssdcache \vdash \sproc{c}{r_t, \sifelse{\strue}{e_1}{e_2}} \to
      \sheap; \sscache; \ssdcache \vdash \sproc{c}{r_t, e_1}
      }
      \\ \\

      \inference[\rcondf]
      {}
      {
      \sheap; \sscache; \ssdcache \vdash \sproc{c}{r_t, \sifelse{\sfalse}{e_1}{e_2}} \to
      \sheap; \sscache; \ssdcache \vdash \sproc{c}{r_t, e_2}
      }
      \\ \\

      \inference[\rred]
      {}
      {
      \sheap; \sscache; \ssdcache \vdash \sproc{c}{r_t, \slet{x}{\tau}{v}{e}} \to
      \sheap; \sscache; \ssdcache \vdash \sproc{c}{r_t, e[v/x]}
      }
      \\ \\



      \inference[\rcall]
      {
      \sheap(r) = C(\overrightarrow{f \mapsto v'}) \\
      m(\overrightarrow{x:\tau})\{\sfmt{return} \; e;\} \in C
      }
      {
      \sheap; \sscache; \ssdcache \vdash \sproc{c}{r_t, r.m(\vec{v})} \to
      \sheap; \sscache; \ssdcache \vdash \sproc{c}{r_t, e[\vec{v}/\vec{x}][r/\sfmt{this}]}
      }
      \\ \\

      \inferencel[\rfield]
      {
      r \in \dom{\sheap} \\
      \neg \vol{v.f} \\\\
      \sscache(r.f) = v \\
      r.f \notin \dom{\ssdcache}
      }
      {
      \sheap; \sscache; \ssdcache \vdash \sproc{c}{r_t, r.f} \toa{R}
      \sheap; \sscache; \ssdcache \vdash \sproc{c}{r_t, v}
      }
      \\ \\

      \inferencel[\rfieldd]
      {
      r \in \dom{\sheap} \\
      \neg \vol{v.f} \\\\
      \ssdcache(r.f) = v
      }
      {
      \sheap; \sscache; \ssdcache \vdash \sproc{c}{r_t, r.f} \toa{R}
      \sheap; \sscache; \ssdcache \vdash \sproc{c}{r_t, v}
      }
      \\ \\

      \inferencel[\rassign]
      {
      v \in \dom{\sheap} \\
      \neg \vol{v.f} \\
      \ssdcache' = \ssdcache[r.f \mapsto v] \\
      }
      {
      \sheap; \sscache; \ssdcache \vdash \sproc{c}{r_t, \sassign{r.f}{v}} \toa{W}
      \sheap; \sscache; \ssdcache' \vdash \sproc{c}{r_t, v}
      }
      \\ \\

      \inferencel[\rnew]
      {
      r - \mathit{fresh} \\
      \sheap(r) = C(\overrightarrow{f \mapsto 0}) \\
      \sclass \in J
      }
      {
      \sheap; \sscache; \ssdcache \vdash \sproc{c}{r_t, \snew{\vec{v}}} \to \\
      \sheap; \sscache; \ssdcache \vdash \sproc{c}{r_t, \slet{\_}{\tunit}{e[\vec{v}/\vec{f}][r/\sfmt{this}]}{r}}
      }
      \\
    \end{array}
  \end{displaymath}
  \caption{Semantics of Local Operations}
  \label{tbl:coresemantics}
\end{figure}

\cref{tbl:defs} presents a summary of the notations we use in the
operational semantics of \lang, along with their definitions.
We discuss these definitions in detail below, together with the
operational semantics.
To improve readability, we split the operational semantics in four
categories:
\emph{core} semantics regarding the core language;
\emph{synchronization} semantics regarding volatile accesses, monitor
handling, join, and interrupts;
semantics for \emph{implicit operations} performed by the JVM;
and \emph{global} semantics regarding parallel execution.

\subsubsection{Core Semantics}
\cref{tbl:coresemantics} presents the \emph{core} semantics of
\lang{}.
Following the notation of Johnsen \em et al.\em, local configurations are
of the form $\sheap; \sscache; \ssdcache \vdash e$.
Note that in the conclusions of some semantic rules we annotate the
$\to$ binary operator with an action kind from JDMM or $\alpha$,
\em e.g.\em,
we use $\toa{R}$ to show that $\rfield$ performs a read action $R$.
In the proof presented in~\cref{sec:proof}, we present all action kinds along with their
abbreviations used in the annotations, and use this information to argue
about the adherence of the operational semantics to JDMM.
Note that $c$ and $r_t$ in $\sproc{c}{r_t, e}$, although present in
every rule, are not involved in any of the rules
in~\cref{tbl:syncsemantics}.
We use them to argue about the global semantics, shown
in~\cref{tbl:globsemantics}.
This syntax allows us to argue about which core is executing a thread
and what is the corresponding object of this thread.

The $\rctxstep$ rule describes the evaluation of an expression in a
context.
The $\rcondt$ and $\rcondf$ rules handle
conditional expressions in the standard manner.
Rule $\rred$ handles substitution in the standard manner.
%
%
Rule $\rcall$ handles method calls.
We use $r.m(\vec{v})$ for invocations with arguments $\vec{v}$ of the
method with name $m$ of the object referenced by $r$.
To determine the body of the method we use
$m(\overrightarrow{x:\tau})\{\sfmt{return} \; e;\}$, where
$\overrightarrow{x:\tau}$ are the formal arguments of the method and $e$
is the method body.
We evaluate method calls by substituting the formal arguments with the
given ones and \texttt{this} with $r$ in the method body.
%

In our VM, all memory accesses first go through the write
buffer; if they miss they proceed to the object cache.  Thus, to access a
field we need it to be present either in the write buffer or the object
cache.
To reason about such accesses we define two structural rules, $\rfield$
and $\rfieldd$.
Rule $\rfield$ handles non-volatile field accesses, when the field is
cached in the object cache, and $\rfieldd$ handles non-volatile field
accesses, when the field is cached in the write buffer.

In $\rfield$, the first premise requires that the object containing
the field being accessed is in the heap (has been allocated and
initialized).
The second premise requires the access
to not refer to a volatile field.  To achieve this we use the function
$\vol{r.f}$ which returns true if the field $f$ is
\texttt{volatile} in the object referenced by $r$ and false otherwise.
This function models
the distinction, performed internally by the JVM, of volatile fields
from normal fields.
The third premise requires that the core performing the read has a local
copy of the field in its object cache, and the cached value is $v$.
The last premise requires that the field is not cached in the write
buffer.
Considering $\sheap$, $\sscache$, and $\ssdcache$ as maps $X$, we use
$X(k)$ to get the value of the cached object or field with key $k$.
We also use $\sscache(r.f) = v$ as a shorter notation of $\sscache(r) =
C(f_1' \mapsto v_1', \ldots, f \mapsto v, \ldots, f_n' \mapsto v_n')$ to show that $f$ maps to
$v$ in the object returned by $\sscache(r)$.
Additionally, we use $\dom{X}$ to get all the map keys, i.e., references
in the case of $\sheap$ and $\sscache$ or field names in the case of
$\ssdcache$.

Similarly, $\rfieldd$ handles field accesses of fields that are
cached in the write buffer.
The only difference from $\rfield$ is that we require $f$ to be cached
in the write buffer and get its value from there instead of the object
cache.

Rule $\rassign$ handles non-volatile field writes, which also go through
the write buffer.
As a result, writes change the contents of the write buffer instead of
the heap, as required by the last two premises.
Given a map $X$, $X' = X \setminus k$ is used to show that $X'$ contains
the same mappings as $X$ except a mapping for key $k$, thus $k \not \in
\dom{X'}$ and $X' \subseteq X$.
Note that we use $\subseteq$ instead of $\subset$, since $k$ might not
be in the map in the first place.

Rule $\rnew$ invokes the constructor of the corresponding $\sclass$ in a
similar manner to $\rcall$.
Rule $\rctxstep$ ensures that the constructor will be evaluated before
the reference $r$ will be assigned to any variable.
This ensures that final fields are initialized before \emph{publishing}
the new object.
Similarly to Johnsen \em et al.\em, we use $C\paren{\vec{v}}$ for instances of
class $C$ with field values $\vec{v}$, i.e., field $f_i$ contains the
value $v_i$.
Note that according to the JMM ``\textit{conceptually every object is
created at the start of the program}''~\cite[\S4.3]{Manson:2005:JMMconf}.
That said, in \lang{} we assume that the object is already present in
the memory, with its fields initialized to the default value, and that
$\rnew$ just invokes the constructor and returns a reference to the
object.
We use $r - \mathit{fresh}$ to show that there is no other reference to
that object already.

\begin{figure}[t]
  \begin{displaymath}
    \small
    \begin{array}{c}
      \multicolumn{1}{l}{
      \fbox{$
      \sheap; \sscache; \ssdcache \vdash \sproc{c}{r_t, e}
      \to
      \sheap; \sscache; \ssdcache \vdash \sproc{c}{r_t, e}
      $}
      }
      \\ \\

      \inferencel[\rfetch]
      {
      \sheap(r) = C(\overrightarrow{f \mapsto v}) \\
      \sscache' = \sscache[r \mapsto \sheap(r)]
      }
      {
      \sheap; \sscache; \ssdcache \vdash \sproc{c}{r_t, e} \toa{F}
      \sheap; \sscache'; \ssdcache \vdash \sproc{c}{r_t, e}
      }

      \\ \\

      \inference[\rwriteback]
      {
      r \in \dom{\sheap} \\
      r \in \dom{\sscache} \\
      \neg \vol{r.f} \\
      r.f \in \dom{\ssdcache} \\
      \sheap' = \sheap[r.f \mapsto \ssdcache(r.f)] \\
      \sscache' = \sscache[r.f \mapsto \ssdcache(r.f)] \\
      \ssdcache' = \ssdcache \setminus r.f \\
      }
      {
      \sheap; \sscache; \ssdcache \vdash \sproc{c}{r_t, e} \toa{B}
      \sheap'; \sscache'; \ssdcache' \vdash \sproc{c}{r_t, e}
      }
      \\ \\

      \inferencel[\rinvalidate]
      {
      r \in \dom{\sscache} \\
      \sscache' = \sscache \setminus r
      }
      {
      \sheap; \sscache; \ssdcache \vdash \sproc{c}{r_t, e} \toa{I}
      \sheap; \sscache'; \ssdcache \vdash \sproc{c}{r_t, e}
      }
      \\ \\

      \inferencel[\rstart]
      {
      \sscache = \emptyset\\
      \ssdcache = \emptyset\\
      \sheap(r_t) = C(\overrightarrow{f \mapsto v}, \sfmt{spawned})\\
      \sheap'(r_t) = C(\overrightarrow{f \mapsto v}, \sfmt{started})
      }
      {
      \sheap; \sscache; \ssdcache \vdash \sproc{c}{r_t, \sfmt{start}} \toa{S}
      \sheap'; \sscache; \ssdcache \vdash \sproc{c}{r_t, r_t.run()}
      }

      \\ \\

      \inferencel[\rfinish]
      {
      \ssdcache = \emptyset\\
      \sheap(r_t) = C(\overrightarrow{f \mapsto v}, \sfmt{started})\\
      \sheap'(r_t) = C(\overrightarrow{f \mapsto v}, \sfmt{finished})
      }
      {
      \sheap; \sscache; \ssdcache \vdash \sproc{c}{r_t, \sunit} \toa{\mathit{Fi}}
      \sheap'; \sscache; \ssdcache \vdash \sproc{c}{r_t, \sunit}
      }

      \\
    \end{array}
  \end{displaymath}
  \caption{Operational Semantics for Implicit Operations}
  \label{tbl:synthsemantics}
\end{figure}

\subsubsection{Semantics of Implicit Operations}
\cref{tbl:synthsemantics} presents the operational semantics for
implicit operations.
These are operations performed implicitly by the virtual machine and do
not map to language expressions.
Rules $\rfetch$, $\rwriteback$, and $\rinvalidate$ handle fetching,
write-back, and invalidation of a cached object, respectively.
Fetching an object requires that it exists in the heap (first and second
premise).
A fetch results in the addition of the object referenced by $r$ in the
object cache $\sscache$.
Writing back a field $r.f$ requires that the object referenced by $r$ is
present in the heap $\sheap$ and the object cache $\sscache$, $r.f$ is
not volatile, and there is a dirty copy of it in the write buffer
$\ssdcache$.
Writing-back a field results in the update of its value both in the heap
$\sheap$ and the object cache $\sscache$.
Invalidating an object's cached copy requires that it is cached.
Note that this does not force that object's fields to not be cached in
the write buffer.
An invalidation results in the removal of the object referenced by $r$
from the object cache, $\sscache$, of the core executing the
invalidation.
%
Rule $\rstart$ enforces the evaluation of the thread start action before
any other action in the thread and ---treating thread start as an
acquire action--- requires the object cache and the write buffer to be
empty on the running core.

Rule $\rfinish$ handles the completion of a thread.
Note that a thread reaches completion when its thread body is equal to
the unit value $\sunit$.
As a release action requires the write buffer to be empty, and changes
the state of the thread to allow joiners to proceed.


\begin{figure}[!bt]
  \begin{displaymath}
    \small
    \begin{array}{c}
      \multicolumn{1}{l}{
      \fbox{$
      \sheap; \sscache; \ssdcache \vdash \sproc{c}{r_t, e}
      \to
      \sheap; \sscache; \ssdcache \vdash \sproc{c}{r_t, e}
      $}
      }
      \\ \\

      \inference[\rvolatilereadl]
      {
      r \in \dom{\sheap} \\
      \vol{r.f} \\\\
      \sheap (r.f.l) = 0 \\
      \sheap' = \sheap[r.f.l \mapsto r_t] \\
      }
      {
      \sheap; \sscache; \ssdcache \vdash \sproc{c}{r_t, r.f} \to
      \sheap'; \sscache; \ssdcache \vdash \sproc{c}{r_t, r.f}
      }
      \qquad

      \inference[\rvolatileread]
      {
      r \in \dom{\sheap} \\
      \sheap (r.f.l) = r_t \\\\
      \sscache = \emptyset \\
      \ssdcache = \emptyset \\\\
      \sheap' = \sheap[r.f.l \mapsto 0] \\
      \sheap(r.f) = v \\
      }
      {
      \sheap; \sscache; \ssdcache \vdash \sproc{c}{r_t, r.f} \toa{Vr}
      \sheap'; \sscache; \ssdcache \vdash \sproc{c}{r_t, v}
      }
      \\ \\

      \inference[\rvolatilewritel]
      {
      r \in \dom{\sheap} \\
      \vol{r.f} \\\\
      \sheap (r.f.l) = 0 \\
      \sheap' = \sheap[r.f.l \mapsto r_t] \\
      }
      {
      \sheap; \sscache; \ssdcache \vdash \sproc{c}{r_t, \sassign{r.f}{v}} \to
      \sheap'; \sscache; \ssdcache \vdash \sproc{c}{r_t, \sassign{r.f}{v}}
      }
      \qquad

      \inference[\rvolatilewrite]
      {
      r \in \dom{\sheap} \\
      \sheap (r.f.l) = r_t \\\\
      \ssdcache = \emptyset \\
      \sheap' = \sheap[r.f \mapsto v][r.f.l \mapsto 0]
      }
      {
      \sheap; \sscache; \ssdcache \vdash \sproc{c}{r_t, \sassign{r.f}{v}} \toa{Vw}
      \sheap'; \sscache; \ssdcache \vdash \sproc{c}{r_t, v}
      }
      \\ \\

      \inference[\rmonitorentera]
      {
      r \in \dom{\sheap} \\
      \sscache = \emptyset \\
      \ssdcache = \emptyset \\
      \sheap(r) = (o, 0) \\
      \sheap' = \sheap[r \mapsto (o, r_t(1))]
      }
      {
      \sheap; \sscache; \ssdcache \vdash \sproc{c}{r_t, r.\sfmt{monitorenter}} \toa{L}
      \sheap'; \sscache; \ssdcache \vdash \sproc{c}{r_t, \sunit}
      }
      \\ \\

      \inference[\rmonitorenterb]
      {
      r \in \dom{\sheap} \\
      \sheap(r) = (o, r_t(n)) \\
      \sheap' = \sheap[r \mapsto (o, r_t(n+1))]
      }
      {
      \sheap; \sscache; \ssdcache \vdash \sproc{c}{r_t, r.\sfmt{monitorenter}} \toa{L}
      \sheap'; \sscache; \ssdcache \vdash \sproc{c}{r_t, \sunit}
      }
      \\ \\

      \inference[\rmonitorexita]
      {
      r \in \dom{\sheap} \\
      \ssdcache = \emptyset \\\\
      \sheap(r) = (o, r_t(1)) \\
      \sheap' = \sheap[r \mapsto (o, 0)]
      }
      {
      \sheap; \sscache; \ssdcache \vdash \sproc{c}{r_t, r.\sfmt{monitorexit}} \toa{U}
      \sheap'; \sscache; \ssdcache \vdash \sproc{c}{r_t, \sunit}
      }
      \qquad

      \inference[\rmonitorexitb]
      {
      r \in \dom{\sheap} \\
      \sheap(r) = (o, r_t(n+2)) \\\\
      \sheap' = \sheap[r \mapsto (o, r_t(n+1))]
      }
      {
      \sheap; \sscache; \ssdcache \vdash \sproc{c}{r_t, r.\sfmt{monitorexit}} \toa{U}
      \sheap'; \sscache; \ssdcache \vdash \sproc{c}{r_t, \sunit}
      }

      \\ \\

      \inferencel[\rjoin]
      {
      \sscache = \emptyset\\
      \ssdcache = \emptyset\\
      \sheap(r_t') = C(\overrightarrow{f \mapsto v}, \sfmt{finished})
      }
      {
      \sheap; \sscache; \ssdcache \vdash \sproc{c}{r_t, r_t'.\sfmt{join()}} \toa{J}
      \sheap; \sscache; \ssdcache \vdash \sproc{c}{r_t, ()}
      }

      \\ \\

      \inference[\rinterrupt]
      {
      \ssdcache = \emptyset\\
      \sheap(r_t') = C(\overrightarrow{f \mapsto v}, \sfmt{started})\\
      \sheap'(r_t') = C(\overrightarrow{f \mapsto v}, \sfmt{interrupted})
      }
      {
      \sheap; \sscache; \ssdcache \vdash \sproc{c}{r_t, r_t'.\sfmt{interrupt()}}
      \toa{\mathit{Ir}}
      \sheap'; \sscache; \ssdcache \vdash \sproc{c}{r_t, ()}
      }

      \\ \\

      \inference[\rinterruptedt]
      {
      \sscache = \emptyset\\
      \ssdcache = \emptyset\\
      \sheap(r_t') = C(\overrightarrow{f \mapsto v}, \sfmt{interrupted})
      }
      {
      \sheap; \sscache; \ssdcache \vdash \sproc{c}{r_t, r_t'.\sfmt{interrupted()}}
      \toa{\mathit{Ird}}
      \sheap; \sscache; \ssdcache \vdash \sproc{c}{r_t, ()}
      }

      \\ \\

      \inference[\rinterruptedf]
      {
      \mathit{state} \neq \sfmt{interrupted}\\
      \sheap(r_t') = C(\overrightarrow{f \mapsto v}, \mathit{state})
      }
      {
      \sheap; \sscache; \ssdcache \vdash \sproc{c}{r_t, r_t'.\sfmt{interrupted()}}
      \toa{}
      \sheap; \sscache; \ssdcache \vdash \sproc{c}{r_t, ()}
      }








      \\
    \end{array}
  \end{displaymath}
  \caption{Semantics of Synchornization Operations}
  \label{tbl:syncsemantics}
\end{figure}

\subsubsection{Semantics of Synchornization Operations}
\cref{tbl:syncsemantics} presents the \emph{synchronization}
operational semantics.
That is, rules about volatile accesses, monitor handling, join, and
interrupts.

Rules $\rvolatilereadl$ and $\rvolatileread$ handle reads of volatiles.
Rules $\rvolatilewritel$ and $\rvolatilewrite$ handle volatile writes.
The combination of $\rvolatilereadl$ and $\rvolatileread$
results in a single \emph{volatile-read}.
The same holds for $\rvolatilewritel$, $\rvolatilewrite$ and the
\emph{volatile-write} action.
Specifically, for each volatile field $r.f$ we assume a synthetic lock
$r.f.l$.
This lock is used to force a total ordering on the accesses to this
variable and guarantee atomicity to the corresponding hardware memory
accesses, as described in~\cref{sec:implementation:volatiles}.
When $r.f.l$ is $0$, it means the volatile variable $r.f$ is not being
accessed by another thread.
Assigning the thread $r_t$ to $r.f.l$ we essentially block other threads
from accessing this volatile variable.
Additionally, volatile accesses are exceptions to the rule that all
accesses go through the cache.
Since volatile reads are \textit{acquire} actions and volatile writes
are \textit{release} actions, before volatile writes, any dirty data in
the corresponding core's cache must be written-back and before volatile
reads, the corresponding core's cache must be invalidated.
We use $\emptyset$ for empty maps.

Rules $\rmonitorentera$ and $\rmonitorenterb$ handle monitor
acquisition; similarly, rules $\rmonitorexita$ and $\rmonitorexitb$ handle
monitor release.  These
rules use $r.l$ ---not to be confused with the synthetic lock
$r.f.l$ of volatile variables--- to represent the implicit monitor
associated with the object with identity $r$.  Our monitor handling is
similar to the lock handling introduced in~\cite{Johnsen2012257}.
The notation $H(r.l) = 0$ dictates that the corresponding monitor is
not acquired by any thread in the system.
$H(r.l) = r_t(n)$ dictates that the corresponding
monitor has been acquired $n$ times by the thread $r_t$.
Rule $\rmonitorentera$ requires that a monitor must be free before its
acquisition.  Rule $\rmonitorenterb$ requires that a monitor is
already owned by some thread before it gets re-entered by that same
thread.  Rules $\rmonitorexita$ and $\rmonitorexitb$ ensure that a
monitor is released only by its owner and the same number of times it
was previously acquired.

In the case of nested monitor acquisition we can avoid invalidating
the object caches and writing-back data at nesting monitor release.  By
definition, nested acquisition of monitors requires that the monitor
is owned by the same thread at any nesting level.  Under that
assumption, any concurrent actions that operate on the cached data
used in the critical section would be the result of a data-race,
meaning that the program is not DRF\@.  In that case, it is not necessary
for any of the corresponding dirty data to become visible, to the threads
performing the racy accesses, at nested monitor releases.
Note that racy accesses are not guaranteed to see the latest write if
the thread executing them did not \emph{synchronize-with} an action that
\emph{happens-after} that write.
Similarly, since the monitor is already owned by the current thread,
there is no need to invalidate its core's cache in order to get the
latest values, since those values are the results of some data-race.
As a result, rules $\rmonitorenterb$ and $\rmonitorexitb$ do
not need any special premises regarding object caches and write buffers.

Rule $\rjoin$ handles invocations to the \texttt{join()} method of a
thread.
Its first two premises require that the object cache and the write
buffer are empty, since join is an acquire action.
The third premise requires the state of the thread object to be
$\sfmt{finished}$, modeling the way a join blocks on the state of a
thread in the JVM implementation.

Rule $\rinterrupt$ handles invocations to the \texttt{interrupt()}
method of a thread.
Its first premise requires that the write buffer is empty, since
interrupt is a release action.
The second and third premises require the state of the thread object to
be $\sfmt{started}$ before the interrupt and $\sfmt{started}$ after it,
modeling the way interrupts are implemented by changing the thread's
state in the JVM implementation or setting a hardware register in the
case of using hardware interrupts.

Rules $\rinterruptedt$ and $\rinterruptedf$ handle invocations to the
\texttt{interrupted()} method of a thread.
Rule $\rinterruptedt$ handles cases where the thread is interrupted.
Its first two premises require that the object cache and write buffer
are empty, since interrupt detection is an acquire action.
The third premise requires the state of the thread object to be
$\sfmt{interrupted}$.

Rule $\rinterruptedf$ handles cases where the thread is not interrupted.
Its premises require the state of the thread object to not be
$\sfmt{interrupted}$, in such cases the invocation is not a
synchronization action so there is no need for flushing the object cache
or the write buffer.

%

\begin{figure}[t]
  \begin{displaymath}
    \small
    \begin{array}{c}
      \multicolumn{1}{l}{
      \fbox{$
      \sheap; \scaches; \ssdcaches \vdash T
      \xrightarrow[\vec{c}]{\vec{\alpha}}
      \sheap; \scaches; \ssdcaches \vdash T
      $}
      }
      \\ \\

      \inferencel[\rlift]
      {
      \scache{c} = \scaches(c) \\
      \sdcache{c} = \ssdcaches(c) \\\\
      \scache{c}' = \scaches'(c) \\
      \sdcache{c}' = \ssdcaches'(c) \\
      \sheap; \scache{c}; \sdcache{c} \vdash \sproc{c}{r_t, e}
      \toa{\alpha}
      \sheap'; \scache{c}'; \sdcache{c}' \vdash \sproc{c}{r_t, e'}
      \\
      \scaches' = \scaches[c \mapsto \scache{c}']\\
      \ssdcaches' = \ssdcaches[c \mapsto \sdcache{c}']\\
      }
      {
      \sheap; \scaches; \ssdcaches \vdash \sproc{c}{r_t, e}
      \xrightarrow[\aset{c}]{\aset{\alpha}}
      \sheap'; \scaches'; \ssdcaches' \vdash \sproc{c}{r_t, e'}
      }
      \\ \\

      \inferencel[\rspawn]
      {
      \sheap(r_{t'}) = C(\overrightarrow{f \mapsto v}) \\
      \sheap'(r_{t'}) = C(\overrightarrow{f \mapsto v}, \sfmt{spawned}) \\
      \sfmt{run}()\{\sfmt{return} \; e;\} \in C \\
      \ssdcaches(c) = \emptyset \\
      c' \in \stids
      }
      {
      \sheap; \scaches; \ssdcaches \vdash \sproc{c}{r_{t}, r_{t'}.\sfmt{start}()}
      \xrightarrow[\aset{c}]{\aset{Sp}}\\
      \sheap'; \scaches; \ssdcaches \vdash \sproc{c}{r_{t}, \sunit} \parallel \sproc{c'}{r_{t'}, \sfmt{start}}
      }
      \\ \\

      \inferencel[\rmigrate]
      {
      c' \in \stids \\
      c \not = c' \\\\
      \ssdcache(c) = \emptyset \\
      \ssdcache(c') = \emptyset \\
      \sscache(c') = \emptyset
      }
      {
      \sheap; \scaches; \ssdcaches \vdash \sproc{c}{r_t, e}
      \xrightarrow[\aset{c}]{\aset{M}}
      \sheap; \scaches; \ssdcaches \vdash \sproc{c'}{r_t, e}
      }
      \\ \\

      \inferencel[\rpar]
      {
      }
      {
      \sheap; \scaches; \ssdcaches \vdash T_1
      \xrightarrow[\emptyset]{\emptyset}
      \sheap; \scaches; \ssdcaches \vdash T_1
      }
      \\ \\




      \inferencel[\rparg]
      {
      \vec{c_1} \cap \vec{c_2} = \emptyset\\
      \scaches_1 = \scaches \downarrow \vec{c_1} \\
      \scaches_2 = \scaches \downarrow \vec{c_2} \\
      \scaches_3 = \scaches \setminus (\scaches_1 \cup \scaches_2)\\\\
      \ssdcaches_1 = \ssdcaches \downarrow \vec{c_1} \\
      \ssdcaches_2 = \ssdcaches \downarrow \vec{c_2} \\
      \ssdcaches_3 = \ssdcaches \setminus (\ssdcaches_1 \cup \ssdcaches_2)
      \\\\
      \sheap; \scaches_1; \ssdcaches \vdash T_1
      \xrightarrow[\vec{c_1}]{\vec{\alpha_1}}
      \sheap'; \scaches_1'; \ssdcaches_1' \vdash T_1'
      \\
      \sheap; \scaches_2; \ssdcaches \vdash T_2
      \xrightarrow[\vec{c_2}]{\vec{\alpha_2}}
      \sheap; \scaches_2'; \ssdcaches_2' \vdash T_2'
      \\
      \scaches' = \scaches_1' \cup \scaches_2' \cup \scaches_3
      \\
      \ssdcaches' = \ssdcaches_1' \cup \ssdcaches_2' \cup \ssdcaches_3
      }
      {
      \sheap; \scaches; \ssdcaches \vdash T_1 \parallel T_2
      \xrightarrow[\vec{c_1} \cup \vec{c_2}]{\vec{\alpha_1} \cup \vec{\alpha_2}}
      \sheap'; \scaches'; \ssdcaches' \vdash T_1' \parallel T_2'
      }

      \\
    \end{array}
  \end{displaymath}
  \caption{Global Operational Semantics}
  \label{tbl:globsemantics}
\end{figure}

\subsubsection{Semantics of Global Operations}
In~\cref{tbl:globsemantics} we present the global operational semantics
of \lang{}.
Similarly to the local configurations, the global configurations are of
the form $\sheap; \scaches; \ssdcaches \vdash e$, where $\scaches$ and
$\ssdcaches$ are all the system's object caches and write buffers
respectively, while $\scaches(c)$ and $\ssdcaches(c)$ are the object
cache and write buffer of core $c$, respectively.
Note that the heap is the same in global and local configurations
since it is shared among all cores.

Rule $\rlift$ lifts local reduction steps to the global level.
We use $\scaches[c \mapsto \scache{c}']$ and
$\ssdcaches[c \mapsto \sdcache{c}']$ to show that the state of
$\scaches(c)$ and $\ssdcaches(c)$ in the system is replaced by
$\scache{c}'$ and $\sdcache{c}'$, respectively.

Rule $\rspawn$ handles thread spawns (i.e., \texttt{Thread.start()}
calls).
For every spawn ---which is also a release action--- we require that all
dirty data are written-back.
Then the JVM picks one of the available cores, marked as $c'$ and
schedules thread $v'$ to it.
We represent this by introducing $\sproc{c'}{r_t', \sfmt{start}}$ in
parallel to the previously running
$\sproc{c}{r_t, r_t'.\sfmt{start()}}$.
Note that $\rspawn$ changes the state of the thread to $\sfmt{started}$
to mark that this thread has started and forbid any re-spawns.

Rule $\rmigrate$ handles the Java thread migration to another core by
the scheduler.
It picks one of the available cores, marked as $c'$ and replaces $c$
with it, representing that thread $r$ will continue its execution on
core $c$ instead of $c'$.

Rule $\rpar$ is essentially a no-op that allows threads to block and not
step in every transition in an execution trace, as \em e.g.\em, a
finished but not joined thread.

In \lang, two (or more) Java threads can step concurrently through the
$\rparg$ rule.
Each thread may change its core's object cache and write buffer state
and thus affect $\scaches$ and $\ssdcaches$.
Since the object caches and write buffers are disjoint for each core,
the resulting global state of object caches and write buffers after a
concurrent step is the union of the changed object buffers and write
buffers by each set of cores that step in the parallel transition and
those that where left unchanged by both.
To get the object caches and write buffers that a set of cores $\vec{c}$
changes we use $\scaches \downarrow \vec{c}$ (projection).
Note that the first premise of $\rparg$ required the two sets of cores
that perform a step in the parallel transition to be disjoint.
This is to model that each core is running a single thread and performs
a single step each time.
Additionally, inspecting its eighth and ninth premise it only allows a
single set of threads to modify the heap.
This limitation partially models the hardware memory bus and how it
orders memory transfers.
We allow only one write per step to the heap, this way we allow
parallelism but not concurrent writes to the heap.
To improve this, one can slice the heap, then different synchronization
managers may handle different slices of the heap and increase
parallelism.


\subsection{Proof Sketch}
\label{sec:shortproof}

This section briefly describes the proof of \lang's adherence to the JDMM.
For a detailed proof of adherence~\cref{sec:proof}.  Intuitively, the correctness property can
be expressed as:
\begin{theorem}\label{t:adherence2}
  \lang{}'s operational semantics generates only well-formed execution
  traces.
\end{theorem}
To prove \autoref{t:adherence2}, we show by induction that \lang's
operational semantics satisfies every well-formedness rule.
That is, given any well formed execution trace:
%
$$
  \sheap; \scaches; \ssdcaches \vdash T_1 \parallel T_2
  \to^* \sheap'; \scaches'; \ssdcaches' \vdash T_1' \parallel T_2'
$$
%
we show that the trace after taking one more step:
%
$$
  \sheap; \scaches; \ssdcaches \vdash T_1 \parallel T_2 \to^*
  \sheap'; \scaches'; \ssdcaches' \vdash T_1' \parallel T_2' \to
  \sheap''; \scaches''; \ssdcaches'' \vdash T_1'' \parallel T_2''
$$
%
is well-formed as well.

This amounts to essentially a preservation proof for each rule, many
of which are straightforward.
It is trivial to show that structural rules with conclusions that do
not affect the memory state and do not regard synchronization actions
preserve the well-formedness of the execution.
For the rest, we argue about their effects on the execution state.
Since \lang's operational semantics is tailored after JDMM's
well-formedness rules, for most inference rules, inspecting their
premises and conclusions is enough to show that a well-formedness rule
is preserved.

As \lang models DiSquawk executions, we claim that DiSquawk executions
adhere to the JDMM, and consequently to the JMM.

\section{Related Work}
\label{sec:related}

To the best of our knowledge, the only other JVM implementing the Java
memory model on a non cache coherent architecture is
Hera-JVM~\cite{McIlroy2010}.
Hera-JVM also employs caches which it handles in a similar manner to our
implementation, with the difference that it starts a write-back at every
write, as we discuss in~\cref{sec:implementation}.
Regarding the synchronization mechanisms, Hera-JVM relies on the Cell
B.E.'s \texttt{GETLLAR} and \texttt{PUTLLC} instructions to build an
atomic compare-and-swap operation.
However, such instructions are not available on the architectures at
hand~\cite{howard2010,lyberis2012formic}.
Additionally, Hera-JVM did not aim to formally prove its adherence to
the JMM.

Contrary to the implementation, language operational semantics are
often used to formalize memory models.
Previous work describes the memory semantics for shared memory multicore
processor architectures, such as Power~\cite{mador-haim2012power},
x86~\cite{owens2009better,sarkar2009x86}, and ARM~\cite{alglave2009arm}
processors, without focusing on a specific language semantics or memory
model.
Sarkar \em et al.\em~\cite{sarkar2012cpp} first combined the semantics of an architecture
with the memory model definition of the C++ language, focusing on its
execution on shared-memory Power processors.
Pratikakis \em et al.\em~\cite{pratikakis2011semantics} similarly present operational semantics
for a specialized task-parallel programming model designed to target
distributed-memory architectures.
Our work differs from the aforementioned in that it is targeting
distributed or non cache coherent memory architectures.

Boudol and Petri~\cite{Boudol:2009:RMM:1480881.1480930} define a relaxed memory model
using an operational semantics for the Core ML language.
Their work takes into account write buffers that must become empty before
a lock release.
Although the handling of write buffers is similar to handling caches
regarding the write backs, the fetching and invalidation handling part
is not covered in that work.
Additionally, the authors only consider lock releases as synchronization
points, while in the Java language there are multiple synchronization
points according to JMM.
Joshi and Prasad~\cite{salil2010} extend the above work and
define an operational semantics that accounts for caches, namely
update and invalidation cache operations not previously supported.
The authors use a simple imperative language, claiming it
has greater applicability.
Unfortunately, this approach further abstracts away details regarding
the correct implementation of a specific programming language's memory
model.
In our work we focus on the Java language and provide all the needed
details for the implementation of its memory model.
Furthermore, both of the above papers define operational semantics for
generic relaxed memory models.  We believe that defining the
operational semantics for a specific memory model, in this case the
JMM, is a different task that focuses on the issues specific to the
Java language.

Demange \em et al.\em~\cite{Demange:2013:PBB:2429069.2429110} present the operational
semantics of BMM, a redefinition of JMM for the TSO memory model.
BMM is similar to this work in that it aims to bring the Java Memory
Model definition closer to the hardware details.
BMM, however, focuses on buffers instead of caches and assumes the TSO
memory model, which is stricter than the memory model of the
non cache coherent architectures at hand.

Jagadeesan \em et al.\em~\cite{jagadeesan2010generative} also describe an operational semantics
for the Java Memory Model.
Their work, however, does not account for caches or buffers.
It abstracts away the hardware details and considers reads and writes to
become actions that float into the evaluation context.
This approach does not explicitly define when and where writes should be
eventually committed to satisfy the JMM\@.
In our approach, we explicitly define where data get stored
after any evaluation step.

We thus consider our approach to be closer to the implementation.
Cenciarelli \em et al.\em~\cite{cenciarelli2007} use a combination of operational, denotational,
and axiomatic semantics to define the JMM\@.
In that work, the authors show that all the generated executions adhere
to the JMM, but as in~\cite{jagadeesan2010generative} they do not account
for the memory hierarchy.


\section{Conclusions}
\label{sec:conclusions}

This paper presents DiSquawk, a Java VM implementation of the Java
Memory Model that targets a 512-core non cache coherent architecture,
and a proof sketch that it adheres to JMM.
We discuss design decisions and present evaluation results from the
execution of a set of benchmarks from the Java Grande
suite~\cite{grande}.
To prove the correctness of our implementation, we model all key
points of the design using a core calculus \lang and its operational semantics.
\lang is a concurrent java calculus aware of software caches and their
mechanisms.
DiSquawk has been developed as part of the GreenVM
project~\cite{greenvm} and is available for download at
\url{https://github.com/CARV-ICS-FORTH/disquawk}.





\bibliographystyle{abbrv}
\bibliography{paper}  


\clearpage
\pagebreak
\appendix

\section{JDMM Formal Definitions}
\label{app:jdmm}

This appendix presents the JDMM's formal definitions and their
corresponding formalism in \lang{}, where appropriate.

\paragraph{Distributed Execution:}
A distributed execution $E_D$ is a tuple:

\noindent{}\hfil$
E_D = \tuple{P, A_D, \poD, \soD, W, V, \mathit{Cs}, \mathit{Bf},
  \mathit{Ab}, \mathit{Ai}, \swD, \hbD}
$

\noindent{}where:

\begin{itemize}

\item The program $P$ is a set of instructions, in \lang{} this is the
  program $J$.

\item $A_D$ is a set of \emph{actions}.

  \paragraph{Actions:} The JMM abstracts thread operations as
  actions~\cite[\S5.1]{MansonThesis}.  An action is a tuple $\tuple{r_t,
    k, r.f, u}$, where $t$ is the thread performing the action; $k$ is the
  kind of action; $v$ is the (runtime) variable, monitor, or thread,
  involved in the action; and $u$ is a unique, among the actions, identifier.


  JDMM uses the following abbreviations to describe all possible kinds
  of actions:
  \begin{itemize}
  \item $R$ for read, $W$ for write, and $\mathit{In}$ for initialization of a
    heap-based variable
  \item $\mathit{Vr}$ for read and $\mathit{Vw}$ for write of a volatile variable
  \item $L$ for the lock and $U$ for the unlock of a monitor
  \item $S$ for the start and $\mathit{Fi}$ for the end of a thread
  \item $\mathit{Ir}$ for the interruption of a thread and
    $\mathit{Ird}$ for detecting such an interruption by another thread
  \item $\mathit{Sp}$ for spawning (\verb!Thread.start()!) and $J$ for joining a
    thread or detecting that it terminated
  \item $E$ for external actions, i.e., I/O operations
  \item $F$ for fetch from heap-based variables,
  \item $B$ for write-backs of heap-based variables,
  \item $I$ for invalidations of cached variables.
  \end{itemize}

  In \lang{} we use $\Sigma$
  $\xrightarrow[\vec{c}]{(c' \mapsto \tuple{r_t,k,r.f,u}) \in \vec{\alpha}}$
  $\Sigma'$ to denote a transition from state $\Sigma$ to state
  $\Sigma'$, where $\vec{c}$ is the set of cores involved in the
  transition and $c'$ is the core performing the JDMM action
  $\tuple{r_t,k,r.f,u}$ in this transition.

  To get the set of actions $A_D$, from a program's \lang{} execution
  trace:

  \hfil$ \Sigma_0$ $ \xrightarrow[\vec{c_1}]{\vec{\alpha_1}} $
  $\Sigma_1 $ $\xrightarrow[\vec{c_2}]{\vec{\alpha_2}} $
  $\Sigma_2 \ldots \Sigma_n $ $\xrightarrow[\vec{c_n}]{\vec{\alpha_n}} $
  $\Sigma_{n+1} $

  we take the union of the ranges $\rng{\vec{\alpha}}$, where
  $\vec{\alpha}$ is a set of mappings from cores to JDMM actions, i.e.:

  \hfil$\overrightarrow{(c \mapsto \tuple{r_t,k,r.f,u})}$

  Formally:
  \hfil$
  A_D = \rng{\vec{\alpha_1}} \cup
  $
  $
  \rng{\vec{\alpha_2}} \cup
  $
  $
  \ldots \cup \rng{\vec{\alpha_n}}
  $

\item The program order $\poD$ is a relation on $A_D$ defining the order
  of actions regarding a single thread $t$ in $A_D$.
  JDMM uses $x \po y$ to show that $x$ comes before $y$ according to the
  program order within a thread.
  Every pair of actions executed by a single thread $t$ are ordered by
  the program order:

  \hfil$
  \big((x \neq y) \wedge
  (x.t = y.t)\big) \Leftrightarrow \big((x \poD y) \vee (y \poD x)\big)
  $\\

\item The synchronization order $\soD$ is a relation on $A_D$ defining a
  global ordering among all \emph{synchronization actions} in $A_D$

  \paragraph{Synchronization Actions:}
  Any actions with kind $\mathit{In}$, $\mathit{Ir}$, $\mathit{Ird}$,
  $\mathit{Vr}$, $\mathit{Vw}$, $L$, $U$, $S$, $\mathit{Fi}$,
  $\mathit{Sp}$, or $J$ are \emph{synchronization actions}, which form
  the only communication mechanism between threads.
  JDMM uses $x \in \sa(A_D)$ to show that $x$ is a synchronization
  action in $A_D$:

  $\sa(A_D) =$
  $\aset{x \in A: x.k \in \aset{\mathit{In}, \mathit{Ir}, $
      $\mathit{Ird}, \mathit{Vr}, \mathit{Vw}, L, U, S, \mathit{Fi}, $
      $\mathit{Sp}, J}, F, B} $

  JDMM uses $x \soD y$ to show that $x$ comes before $y$ according to
  the synchronization order.
  Every pair of synchronization actions are ordered by synchronization
  order:

  \hfil$ x.k, y.k \in \sa(A_D) \Leftrightarrow \big((x \soD y) \vee (y \soD
  x)\big) $\\

  In \lang{} we group syncrhonization actions of the kinds
  $\mathit{Ird}$, and $J$ in the acquire actions family, denoted by
  $\mathit{Acq}$.
  We also group syncrhonization actions of the kinds $\mathit{Fi}$,
  $\mathit{Ir}$, and $\mathit{In}$ in the release actions family,
  denoted by $\mathit{Rel}$.

  As a result, in \lang{}:

  $\sa(A_D) =$
  $\aset{x \in A: x.k \in \aset{\mathit{Vr}, \mathit{Vw}, L, U, S, $
      $\mathit{Sp}, F, B, \mathit{Acq}, \mathit{Rel}}} $

\item The \emph{write-seen} function $W$ for every read action $r$
  returns the write action seen by $r$, in $A_d$.
  As a result, $W(r).v = r.v$.








\item The \emph{value-written} function $V$ returns the value written by
  every write action $w$, in $A_D$.
  As a result, every read $r$, in $A_D$, reads the value
  $V\big(W(r)\big)$.







\item The \emph{cache-action-seen} function $\mathit{Cs}$ returns the
  fetch or write action seen by any read $r$, in $A_D$.
  Note that: $\mathit{Cs}(r) \poD r$ and
  $\mathit{Cs}(r).k \in \aset{W,F}$.













\item The \emph{write-back-fetched} function $\mathit{Bf}$ returns the
  write-back action whose data each fetch action fetches, in $A_D$.








\item The \emph{action-written-back} function $\mathit{Ab}$ returns the
  write action whose data each write-back writes-back, in $A_D$.
  Note that:\\
  $\mathit{Ab}(b) \poD b$ and $\mathit{Ab}(b).k \in \aset{In, W, Vw}$.\\

  In \lang{}, $\mathit{Ab}(\tuple{r_t, B, r.f, u'})$ returns the
  initialization or write action
  $\tuple{r_t, \mathit{In}\text{ or }W, r.f, u}$ whose data
  $\tuple{r_t, B, r.f, u'}$ writes-back, according to the execution
  trace.
  Note, that in \lang{} we exclude volatile writes from the possible
  kind of actions returned by $\mathit{Ab}$, since volatile writes are
  never written-back by a separate write-back action, they are
  immediately written to the heap.






\item The \textit{action-invalidated} function $\mathit{Ai}$, returns
  the write or fetch action that cached the data invalidated by each
  invalidation action, in $A_D$.
  Note that:
  $\mathit{Ai}(i) \poD p$ and $\mathit{Ai}(i).k \in \aset{W,F}$.\\

  In \lang{}, $\mathit{Ad}(\tuple{r_t, I, r.f, u'})$ returns the
  write-back or fetch action $\tuple{r_t, W\text{ or }F, r.f, u}$
  writing or fetching a value $v_w$ that $\tuple{r_t, I, r.f, u'}$
  invalidates, according to the execution trace.
  Note that in \lang{} instead of write actions the function returns
  write-back actions, since write actions update the write buffer, which
  cannot be invalidated, and write-back actions update the values in the
  object cache, removing the corresponding entries from the write buffer.







\item The distributed synchronizes-with order $\swD$ is a relation on
  $A_D$ defining which actions in $A_D$ synchronize with each other.\\

  JDMM uses $x \swD y$ to show that $x$ synchronizes-with $y$.
  Note that $x \swD y \Rightarrow x \so y$.
  An action $x$ synchronizes-with an action $y$, written $x \swD y$,
  when:
  \begin{itemize}

  \item $x$ is the initialization of variable $v$ and $y$ is the first
    action of any thread:

    \hfil$\big((x.k = In) \wedge (y.k =S)\big) $

  \item $y$ is a subsequent read of the volatile variable written by
    $x$:

    \hfil$ (x.k = \mathit{Vw}) \wedge (y.k = \mathit{Vr}) \wedge (x \so y) $

  \item $y$ is a subsequent lock of the monitor that $x$ unlocked:

    \hfil$ (x.k = U) \wedge (y.k = L) \wedge (x.v = y.v) \wedge (x \so y) $

  \item $y$ is the start action of thread $t$ and $x$ is the spawn of
    $t$:

    \hfil$ (x.k = Sp) \wedge (y.k = S) \wedge (x.v = y.t) $

  \item $y$ is a call to \verb!Thread.join()! or \verb!Thread.isAlive()!
    and $x$ is the finish action of this thread:

    \hfil$ (x.k = Fi) \wedge (y.k = J) \wedge (x.t = y.v) $

  \item $y$ is an action detecting if a thread has been interrupted and
    $x$ is an interrupt to that thread:

    \hfil$ (x.k = Ir) \wedge (y.k = Ird) \wedge (x.v = y.v) $

  \item $y$ is the implicit read of a reference to the object being
    finalized and $x$ is the end of the constructor of this object.


  \end{itemize}
  In the synchronizes-with examples above, when comparing the variable
  $v$ of one action with the thread $t$ of the other (i.e., $x.t = y.v$)
  means that $y$ acts on thread $x.t$.  The $x$ action is a
  \emph{release} action and $y$ is an \emph{acquire} action.  A
  \emph{release} action must make all writes, visible to the executing
  thread, visible to the actions following (according to any of the
  orders defined till now) the \emph{acquire} action.\\

  In \lang{}, given any execution trace:

  $ \ldots \Sigma_1$
  $ \xrightarrow[\_]{\vec{\alpha_1}: \tuple{\_,k,r.f,u} \in \vec{\alpha_1}} $
  $ \Sigma_2 \ldots \Sigma_{n-1}$
  $ \xrightarrow[\_]{\vec{\alpha_n}: \tuple{\_,k',r.f,u'} \in \vec{\alpha_n}}$
  $ \Sigma_n \ldots $

  \noindent{}where $k, k' \in \sa(A_D)$, \emph{if and only if} $k$ and
  $k'$ can form a synchronization pair and there is no other transition:

  \hfil$\Sigma_x \xrightarrow[\_]{\vec{\alpha_x}: \tuple{\_,k\text{ or }k',r.f,u''} \in \vec{\alpha_x}} \Sigma_y$

  between the transitions that contain the actions with id $u$ and $u'$
  then:

  \hfil$\tuple{\_,k,r.f,u} \swD \tuple{\_,k',r.f,u'}$

\item The happens-before order $\hbD$ is a relation on $A_D$ that
  defines a partial order among actions in $A_D$.\\

  The happens-before notion is the one introduced by Lamport
  in~\cite{lamport1978time}.
  In the context of the JMM this is the transitive closure of the
  program order and the synchronizes-with order.
  JDMM uses $x \hbD y$ to show that $x$ happens-before $y$.

  In \lang{}, given any execution trace:

  $
  \ldots \Sigma_1
  \xrightarrow[\_]{\vec{\alpha_1}: \tuple{\_,\_,\_,u} \in \vec{\alpha_1}}
  \Sigma_2 \ldots
  $
  $
  \Sigma_{n-1}
  \xrightarrow[\_]{\vec{\alpha_n}: \tuple{\_,\_,\_,u'} \in \vec{\alpha_n}}
  \Sigma_n \ldots
  $

  if any of the following holds:

  \begin{itemize}
  \item $\tuple{\_,\_,\_,u} \poD \tuple{\_,\_.\_,u'}$
  \item $\tuple{\_,\_,\_,u} \swD \tuple{\_,\_.\_,u'}$
  \item there exists a transition $ \Sigma_x $
    $ \xrightarrow[\_]{\vec{\alpha_x}: \tuple{\_,\_.\_,u''} \in \vec{\alpha_x}} $
    $ \Sigma_y $ that appears between the transitions that contain the
    actions with ids $u$ and $u'$, in the execution trace, and

    \hfil$\tuple{\_,\_,\_,u} \hbD \tuple{\_,\_.\_,u''} \hbD \tuple{\_,\_.\_,u'}$

    (transitivity)
  \end{itemize}

  then $\tuple{\_,\_,\_,u} \hbD \tuple{\_,\_.\_,u'}$.
\end{itemize}

\paragraph{Conflicting Accesses:} If one of two accesses to the same
variable is a write then these two accesses are \emph{conflicting}.

\paragraph{Data-Race:} A data-race occurs when two conflicting
accesses may happen in parallel.  That is, they are not ordered by
happens-before.

\paragraph{Correctly Synchronized or Data-Race-Free Program:}~\\
A program is correctly synchronized or DRF if and only if all
sequentially consistent executions are free of data-races.

\paragraph{Well-Formed Distributed Execution:}~\\
JDMM defines well-formed executions similarly to the JMM\@.
%
Specifically, in JDMM, a distributed execution $E_D$ is well-formed when:
\begin{description}
  \setlength{\itemsep}{1em}

\item[WF-1] \phantomsection\label{app:wf1}
  Each read of a variable $v$ sees a write to $v$:

  \hfil$ \forall r \in A_D: \exists y \in A_D: \big(W(r) = y\big) $

  Note that the original formal definition in
  JDMM~\cite[\S3]{zakkak:jdmm} is:

  \hfil$ \forall x \in A_D: (x.k = R) \Rightarrow \exists y \in A_D:
  \big(W(x) = y\big) $

  where volatile reads are not considered.
  However, JMM~\cite[\S4.4]{Manson:2005:JMMconf} states that
  ``\textit{For all reads $r \in A$, we have $W(r) \in A$ and
    $W(r).v = r.v$. The variable $r.v$ is volatile if and only if $r$ is
    a volatile read, and the variable $w.v$ is volatile if and only if
    $w$ is a volatile write.
  }'', where to our understanding $w$ refers to $W(r)$, and $r$ refers
  to both volatile and non-volatile reads.
  As a result, in this work, we chose to take volatile reads into
  account as well.
  \\

  In \lang{}, this means that given the execution trace of
  $E_D$, for every transition containing a read action:

  \hfil $\Sigma$ $\xrightarrow[\_]{\vec{\alpha}: \tuple{\_, R\text{ or
      }\mathit{Vr}, r.f, \_} \in \rng{\vec{\alpha}}}$ $ \Sigma'$

  in that trace, there is at least one transition containing a write or
  initialization action:

  \hfil$\Sigma_x$
  $\xrightarrow[\_]{\vec{\alpha'}: \tuple{\_, \mathit{In}\text{ or
      }W\text{ or }\mathit{Vw}, r.f, \_} \in \rng{\vec{\alpha'}}}$
  $\Sigma_y$

  which writes in $r.f$ the value that this read action sees.

\item[WF-2] \phantomsection\label{app:wf2}
  All reads and writes of volatile variables are volatile actions:

  $ \forall x \in A_D: x.k \in \aset{\mathit{Vw},\mathit{Vr}}
  \Rightarrow$
  $\nexists y \in A_D: (y.k \in \aset{R, W}) \wedge (x.v = y.v) $\\

  In \lang{}, this means that given the execution trace of $E_D$, in
  every transition $\Sigma \xrightarrow[\_]{\vec{\alpha}} \Sigma'$ for
  every action

  \hfil$\tuple{\_, k, r.f, \_} \in \rng{\vec{\alpha}}$

  $k$ is either $\mathit{Vr}$ or $\mathit{Vw}$, \emph{if and only if}
  $r.f$ is a volatile variable.

\item[WF-3] \phantomsection\label{app:wf3}
  The number of synchronization actions preceding another
  synchronization action $y$ is finite:

  \hfil$ \forall y \in \sa(A_D): \#\{x \in \sa(A_D): x \soD y\} < \infty $\\




\item[WF-4] \phantomsection\label{app:wf4}
  Synchronization order is consistent with program order:

  $ \forall x,y,z \in A_D: $
  $ \big((x.t = z.t) \wedge (x \soD y \soD z)\big) \Rightarrow$
  $(x \poD z) $\\

  In \lang{} this means that given the execution trace of $E_D$, if it
  contains a trace:

  $\ldots \Sigma_1$
  $\xrightarrow[\_]{\vec{\alpha_1}: \tuple{r_t, k_1, \_, u_1} \in \rng{\vec{\alpha_1}}}$
  $\Sigma_2$
  $\xrightarrow[\_]{\vec{\alpha_2}: \tuple{r_t', k_2, \_, u_2} \in \rng{\vec{\alpha_2}}}$
  $\Sigma_3\ldots \Sigma_n$
  $\xrightarrow[\_]{\vec{\alpha_n}: \tuple{r_t, k_n, \_, u_n} \in \rng{\vec{\alpha_n}}}$
  $\Sigma_{n+1} \ldots$

  \noindent{}where $k_1, k_2, k_n \in \sa(A_D)$ and consequently

  \hfil$\tuple{r_t, k_1, \_, u_1} \soD$
  $\tuple{r_t', k_2, \_, u_2} \soD$
  $\tuple{r_t, k_n, \_, u_n}$

  then it cannot also contain the trace:

  $ \ldots \Sigma_n $
  $ \xrightarrow[\_]{\vec{\alpha_n}: \tuple{r_t, k_n, \_, u_n} \in \rng{\vec{\alpha_n}}} $
  $ \Sigma_{n+1} \ldots \Sigma_1 $
  $ \xrightarrow[\_]{\vec{\alpha_1}: \tuple{r_t, k_1, \_, u_1} \in \rng{\vec{\alpha_1}}} $
  $ \Sigma_2 \ldots$

  \noindent{}where
  $\tuple{r_t, k_n, \_, u_n} \poD \tuple{r_t, k_1, \_, u_1}$.

\item[WF-5] \phantomsection\label{app:wf5}
  Lock operations are consistent with mutual exclusion.\\

  The number of lock actions performed on the monitor $m$ by any thread
  $t'$ before, according to the synchronization order, the lock action
  $l$ performed by thread $t$ on the monitor $m$ must be equal to the
  number of unlock actions performed by thread $t'$ before $l$ on the
  monitor $m$:

  $ \forall x \in A_D : \forall t \in T: (x.k = L) \wedge (x.t \neq t)
  \Rightarrow $

  \hfil $ \#\{y \in A_D: (y.t = t) \wedge (y.k = L) \wedge$
  $(y.v = x.v)
  \wedge (y \soD x)\} = $

  \hfill $ \#\{z \in A_D: (z.t = t) \wedge (z.k = U) \wedge$
  $(z.v = x.v) \wedge
  (y \soD x)\} $

  where $T$ is the set of all the execution threads:

  \hfil$T = \{r_t: (\exists x \in A_D: t = x.t)\}$\\

  In \lang{}, this means that given the execution trace of $E_D$, if a
  transition containing a lock acquisition action for a monitor $r.l$:

  \hfil$\Sigma_x$
  $ \xrightarrow[\_]{\vec{\alpha}: \tuple{r_t, L, r.l, u} \in \vec{\alpha}} $
  $\Sigma_y$

  exists in the trace, then for every thread $r_t'$, where
  $r_t' \neq r_t$ the number of transitions containing a lock
  acquisition action for $r.l$:

  \hfil$\Sigma_L$
  $ \xrightarrow[\_]{\vec{\alpha'}: \tuple{r_t', L, r.l, u'} \in \vec{\alpha'}}$
  $ \Sigma_L'$

  which appear earlier in the trace:

  \hfil$\tuple{r_t', L, r.l, u'} \soD \tuple{r_t, L, r.l, u}$

  is equal to the number of transitions containing a lock release action
  for $r.l$:

  \hfil$\Sigma_U$
  $ \xrightarrow[\_]{\vec{\alpha''}: \tuple{r_t', U, r.l, u''} \in \vec{\alpha''}}$
  $ \Sigma_U'$

  that also appear earlier in the trace:

  \hfil$\tuple{r_t', U, r.l, u''} \soD \tuple{r_t, L, r.l, u}$

\item[WF-6] \phantomsection\label{app:wf6}
  The execution obeys intra-thread consistency.

  \noindent{}In \lang{} this means that given the execution trace of
  $E_D$, for every trace:

  $ \ldots \Sigma_1$
  $ \xrightarrow[\_]{\vec{\alpha}: \tuple{r_t, \mathit{In}\text{ or
      }W\text{ or }\mathit{Vw}, r.f, u} \in \rng{\vec{\alpha}}}$
  $ \Sigma_2 \ldots \Sigma_n$
  $ \xrightarrow[\_]{\vec{\alpha'}: \tuple{r_t, R\text{ or
      }\mathit{Vr}, r.f, u'} \in \rng{\vec{\alpha'}}}$
  $ \Sigma_{n+1} \ldots $

  in it, the read action with id $u'$ may return the value written by
  the action with id $u$, \emph{if and only if} between the two
  transitions, performed by thread $r_t$, there is no other transition,
  performed by thread $r_t$, that includes a write action that acts on
  the same variable $r.f$

  \hfil$\Sigma_x$
  $ \xrightarrow[\_]{\vec{\alpha''}: \tuple{r_t, \mathit{In}\text{ or
      }W\text{ or }\mathit{Vw}, r.f, u''} \in \rng{\vec{\alpha''}}}$
  $ \Sigma_y$

\item[WF-7] \phantomsection\label{app:wf7}
  The execution obeys synchronization order consistency.\\

  JMM states that ``\textit{Synchronization order consistency says that
    (i) synchronization order is consistent with program order and (ii)
    each read $r$ of a volatile variable $v$ sees the last write to $v$
    to come before it in the synchronization
    order}''~\cite[\S3.2]{Manson:2005:JMMconf}.
  The first condition is satisfied \emph{if and only if} \wf{4} is
  satisfied, so JDMM examines only the second condition in \wf{7}.

  $ \forall r \in A_D: (r.k = \mathit{Vr}) \Rightarrow $

  \hfil $\Big(\neg \big(r \soD W(r)\big) \wedge \nexists w' \in A_D:
  (w'.k = \mathit{Vw}) \wedge $

  \hfill $(w'.v = r.v) \wedge \big(W(r) \soD w' \soD r\big)\Big)$\\

  In \lang{} this means that given the execution trace of $E_D$, for
  every trace:

  $ \ldots \Sigma_1$
  $ \xrightarrow[\_]{\vec{\alpha}: \tuple{\_, \mathit{Vw}, r.f, u} \in \rng{\vec{\alpha}}} $
  $ \Sigma_2 \ldots \Sigma_n $
  $ \xrightarrow[\_]{\vec{\alpha'}: \tuple{\_, \mathit{Vr}, r.f, u'} \in \rng{\vec{\alpha'}}} $
  $ \Sigma_{n+1} \ldots $

  in it, the volatile read action with id $u'$ returns the value written
  by the volatile write action with id $u$, \emph{if and only if}
  between the two transitions there is no other transition that includes
  a volatile write action that acts on the same variable $r.f$

  \hfil$ \Sigma_x$
  $ \xrightarrow[\_]{\vec{\alpha''}: \tuple{\_, \mathit{Vw}, r.f, u''} \in \rng{\vec{\alpha''}}} $
  $ \Sigma_y $,

\item[WF-8] \phantomsection\label{app:wf8}
  The execution obeys happens-before consistency:

  $ \forall r \in A_D: \Big(\neg \big(r \hbD W(r)\big) \wedge$
  $ \nexists w' \in A_D: (w'.v = r.v) \wedge $
  $ \big(W(r) \hbD w' \hbD r\big)\Big) $\\

  In \lang{} this means that given the execution trace of $E_D$, for
  every trace:

  $ \ldots \Sigma_1$
  $ \xrightarrow[\_]{\vec{\alpha}: \tuple{\_, \mathit{In}\text{ or
      }W\text{ or }\mathit{Vw}, r.f, u} \in \rng{\vec{\alpha}}}$
  $ \Sigma_2 \ldots \Sigma_n$
  $ \xrightarrow[\_]{\vec{\alpha'}: \tuple{\_, R\text{ or
      }\mathit{Vr}, r.f, u'} \in \rng{\vec{\alpha'}}}$
  $ \Sigma_{n+1} \ldots $

  in it, where

  \hfil$\tuple{\_, \mathit{In}\text{ or }W\text{ or
    }\mathit{Vw}, r.f, u} \hbD \tuple{\_, R\text{ or
    }\mathit{Vr}, r.f, u'}$

  the read action with id $u'$ may return the value written by the
  action with id $u$, \emph{if and only if} there is no other
  transition, between the two transitions, that is ordered with them by
  happens-before and includes a write action that acts on the same
  variable $r.f$:

  \hfil$\Sigma_x$
  $ \xrightarrow[\_]{\vec{\alpha''}: \tuple{r_t, \mathit{In}\text{ or
      }W\text{ or }\mathit{Vw}, r.f, u''} \in \rng{\vec{\alpha''}}}$
  $ \Sigma_y$

  where $\tuple{\_, \mathit{In}\text{ or }W\text{ or
    }\mathit{Vw}, r.f, u} \hbD $
  $\tuple{\_, \mathit{In}\text{ or }W\text{ or
    }\mathit{Vw}, r.f, u''} \hbD \tuple{\_, R\text{ or
    }\mathit{Vr}, r.f, u'}$

\item[WF-9] \phantomsection\label{app:wf9}
  Every thread's start action happens-before its other actions
  except for initialization actions:

  $
  \forall x,y,z \in A_D:
  $
  $\big((z.k \not\in \{S,\mathit{In}\}) \wedge (x.k = \mathit{In})
  \wedge (y.k = S)\big) \Rightarrow$
  $(x \hbD y \hbD z $\\

  JMM states that``\textit{The write of the default value (zero, false
    or null) to each variable synchronizes-with to the first action in
    every thread. Although it may seem a little strange to write a
    default value to a variable before the object containing the
    variable is allocated, conceptually every object is created at the
    start of the program with its default initialized
    values. Consequently, the default initialization of any object
    happens-before any other actions (other than default writes) of a
    program.}''~\cite[\S4.3]{Manson:2005:JMMconf}

  \noindent{}As a result, in \lang{} we assume that in the starting
  state of a program's execution trace all the variables used in that
  trace are already initialized and written back to the main memory,
  i.e, all of them fit in the memory and are initialized to zero.
  Since in this work we do not examine allocation techniques and garbage
  collection, this assumption does not interfere with our
  implementation's proof of adherence to JDMM.
  We essentially model a JVM that initializes the heap at boot and does
  not perform any garbage collections during the execution, which is
  actually how our JVM works when garbage collection is turned off.
  To be consistent with the JDMM requirements about the ordering of
  initialization actions we define the beginning of every execution
  trace in \lang{} to be $\Sigma_{init} \to^* \Sigma_{init}'$, where
  $\to^*$ contains only transitions performing the initialization
  actions and their write-backs, for every variable in the execution
  trace, and $\Sigma_{init} \to^* \Sigma_{init}'$ is well-formed ---each
  initialization happens-before its write-back.

\item[WF-10] \phantomsection\label{app:wf10}
  Every read is preceded by a write or fetch action,
  acting on the same variable as the read.\\

  In JDMM all reads of heap-based variables see cached values.
  Formally:

  $\forall r \in A_D: \Big(\big(W(r) \poD r\big) \vee$
  $\exists f \in A_D: $
  $\big((f.v = r.v) \wedge (f \poD r)\big)\Big) $.

  Note that JDMM does not consider simultaneous multithreading and
  context switching in the core model, thus it does not support cache
  sharing in its formal rules~\cite[\S4.2]{zakkak:jdmm}.
  As a result it requires for the read action that sees a value written
  or fetched by another action to be ordered with the latter according
  to program order.
  Cache sharing, however, is examined in~\cite[\S5.2]{zakkak:jdmm} and
  is shown to be safe under JDMM and not break the execution's
  well-formedness if enabled.\\

  In \lang{}, which supports simultaneous multithreading with shared
  caches, this means that given the execution trace of $E_D$, for every
  transition $ \Sigma_R$
  $ \xrightarrow[\_]{\vec{\alpha}: (c \mapsto \tuple{\_, R, r.f, u}) \in {\vec{\alpha}}}$
  $ \Sigma_R'$, there is at least one transition $\Sigma $
  $ \xrightarrow[\_]{\vec{\alpha}: (c \mapsto \tuple{\_, W\text{ or
      }F, r.f, u'}) \in {\vec{\alpha}}}$ $ \Sigma'$ earlier in that
  trace as well, which essentially means that every read performed by a
  core $c$ is preceded by a write or fetch action, also performed by
  $c$, acting on the same variable as the read.

  Note that in the \lang{} definition of \wf{10} we do not include
  volatile accesses.
  This is justified by the fact that in \lang{} volatile reads access
  the heap directly, which can be seen as fetching, reading, and
  invalidating the variable in a single step.
  As a result, in \lang{} there is no other action before a volatile
  read that caches the variable.
  However, we still comply to the JDMM since we conceptually pack a
  fetch in the volatile read itself, meaning that every volatile read is
  indeed preceded by a (conceptual) fetch.

\item[WF-11] \phantomsection\label{app:wf11}
  There is no invalidation, update, or overwrite of a variable's
  cached value between the action that cached it and the read that
  sees it.  Formally:

  $\forall r \in A_D: \nexists x \in A_D: $
  $\Big((x.k \in \{I, F, W\})
  \wedge \big(\mathit{Cs}(r) \poD x \poD r\big)\Big) $\\

  In \lang{}, this means that given the execution trace of $E_D$, for
  every trace:

  $ \ldots \Sigma_1 $
  $\xrightarrow[\_]{\vec{\alpha_1}: (c \mapsto \tuple{\_, W\text{ or
      }F, r.f, u}) \in {\vec{\alpha_1}}} $
  $ \Sigma_2 \ldots \Sigma_n $
  $ \xrightarrow[\_]{\vec{\alpha_n}: (c \mapsto \tuple{\_, R, r.f, u'}) \in {\vec{\alpha_n}}} $
  $ \Sigma_{n+1} \ldots $

  in it, if the read action with id $u'$ sees the value written or
  fetched by the action with id $u$, then there is no other transition
  $\Sigma$
  $ \xrightarrow[\_]{\vec{\alpha}: (c \mapsto \tuple{r_t, I\text{ or
      }F\text{ or}W, r.f, u''}) \in {\vec{\alpha}}}$ $ \Sigma'$
  between the transitions that contain the actions with ids $u$ and
  $u'$.

  Note that, as we explain for \wf{10}, we do not take in account
  volatile accesses and do not require a program order between the
  actions, instead we require that the actions are performed by the same
  core $c$.

\item[WF-12] \phantomsection\label{app:wf12}
  Fetch actions are preceded by at least one write-back of
  the corresponding variable.\\

  For a value to be fetched, it must first be written to the main
  memory.
  The only way to write to the main memory, by definition, is through a
  write-back.
  Formally:

  \hfil$
  \forall f \in A_D, \exists b \in A_D: \big(b = \mathit{Bf}(f)\big)
  $



\item[WF-13] \phantomsection\label{app:wf13}
  Write-back actions are preceded by at least one write to
  the corresponding variable.\\

  For a variable to be written-back, it must be dirty in some cache; a
  cached copy becomes dirty only when written.
  Formally:

  \hfil$
  \forall b \in A_D, \exists w \in A_D: \big(w = \mathit{Ab}(b)\big)
  $\\

  In \lang{} this means that given the execution trace of $E_D$, for
  every transition $\Sigma$
  $ \xrightarrow[\_]{\vec{\alpha}: (c \mapsto \tuple{\_, B, r.f, u}) \in {\vec{\alpha}}}$
  $ \Sigma'$, in it, there is a at least one transition $\Sigma_w$
  $\xrightarrow[\_]{\vec{\alpha'}: (c \mapsto \tuple{\_, W, r.f, u'}) \in \vec{\alpha'}}$
  $\Sigma_w'$ earlier in that trace as well.

\item[WF-14] \phantomsection\label{app:wf14}
  There are no other writes to the same variable between a
  write and its write-back.  Formally:

  $ \forall b \in A_D: $
  $\Big(\nexists w` \in A_D: \big((w'.v. = b.v)
  \wedge (\mathit{Ab}(b) \poD w' \poD b)\big)\Big)$

  \noindent{}In \lang{} this means that given the execution trace of
  $E_D$, for every trace:

  $ \ldots \Sigma_1$
  $ \xrightarrow[\_]{\vec{\alpha}: (c \mapsto \tuple{r_t, W, r.f, u}) \in \vec{\alpha}} $
  $ \Sigma_2 \ldots \Sigma_n $
  $ \xrightarrow[\_]{\vec{\alpha'}: (c \mapsto \tuple{r_t, B, r.f, u'}) \in \vec{\alpha'}} $
  $\Sigma_{n+1} \ldots $

  in it, the write-back action with id $u'$ writes back the value
  written by the action with id $u$, \emph{if and only if} there is no
  other transition containing a write $\Sigma $
  $\xrightarrow[\_]{\vec{\alpha_w}: (c \mapsto \tuple{r_t, W, r.f, u''}) \in \rng{\vec{\alpha_w'}}}$
  $ \Sigma'$ between the transitions that contain the actions with ids
  $u$ and $u'$.

  Note that, as in \wf{10} and \wf{11}, we do not take in account
  volatile accesses and do not require a program order between the
  actions, instead we require that the actions are performed by the same
  core $c$.

\item[WF-15] \phantomsection\label{app:wf15}
  Only cached variables are invalidated.\\

  Invalid cached data cannot be invalidated.
  Formally:

  $ \forall p \in A_D:\nexists p' \in A_D: $
  $\Big(\big(\mathit{Ai}(p) =
  \mathit{Ai}(p')\big)\wedge \big(\mathit{Ai}(p) \poD p' \poD
  p\big)\Big) $\\

  In \lang{} this means that given the execution trace of $E_D$,
  transitions containing invalidation actions:

  \hfil$\sheap; \scaches; \ssdcaches \vdash T $
  $\xrightarrow[\_]{\vec{\alpha}: (c \mapsto \tuple{r_t, I, r.f, u} \in \rng{\vec{\alpha}}} $
  $\sheap'; \scaches'; \ssdcaches' \vdash T'$

  appear in the trace only when $r.f \in \dom{\scaches(c)}$.

\item[WF-16] \phantomsection\label{app:wf16}
  Reads that see writes performed by other threads are
  preceded by a fetch action that fetches the write-back of the
  corresponding write and there is no other write-back of the
  corresponding variable happening between the write-back and the
  fetch.\\

  Since all writes go through the cache, for a write to be seen by a
  read on a different thread, there must exist a write-back action and a
  subsequent fetch action for it.
  Formally:

  $ \forall r \in A_D:\big(W(r).t \neq r.t\big) \Rightarrow \exists
  b,f \in A_D: $

  $\qquad \Big(\big(\mathit{Ab}(b) = W(r)\big) \wedge
  \big(\mathit{Bf}(f) = b \big) \wedge $
  $\big(W(r) \poD b \swD f \poD r\big) \wedge $

  \hfill $\big(\nexists b': (b'.v = b.v) \wedge (b \hb b' \hb r)\big)\Big) $\\

  In \lang{}, which supports simultaneous multithreading with shared
  caches, \wf{16} essentially translates to ``Reads that see writes
  performed by other \emph{cores} are preceded by a fetch action that
  fetches the write-back of the corresponding write and there is no
  other write-back of the corresponding variable happening between the
  write-back and the fetch''

  This means that given the execution trace of $E_D$, for every trace:

  $ \ldots \Sigma_1 $
  $\xrightarrow[\_]{\vec{\alpha}: (c \mapsto \tuple{r_t, W, r.f, u}) \in \vec{\alpha}} $
  $\Sigma_2 \ldots \Sigma_n $
  $\xrightarrow[\_]{\vec{\alpha'}: (c' \mapsto \tuple{r_t', R, r.f, u'} \in \vec{\alpha'}}$
  $ \Sigma_{n+1} \ldots $

  in it, where $c \neq c'$, the read action with id $u'$ may see the
  value written by the action with id $u$, \emph{if and only if} all of
  the following hold:

  \begin{enumerate}
  \item There is a transition containing a fetch action:

    \hfil$\Sigma_f$
    $\xrightarrow[\_]{\vec{\alpha_f}: (c' \mapsto \tuple{r_t', F, r.f, u_f}) \in \vec{\alpha_f}}$
    $\Sigma_f'$

    between the transitions that contain the actions with ids $u$ and
    $u'$,
  \item There is a transition containing a write-back action:

    \hfil$\Sigma_b$
    $\xrightarrow[\_]{\vec{\alpha_b}: (c \mapsto \tuple{r_t, B, r.f, u_b}) \in \vec{\alpha_b}}$
    $\Sigma_b$

    between the transitions that contain the actions with ids $u$ and
    $u_F$,
  \item There is no other transition containing a write-back action:

    \hfil$\Sigma_b'$
    $\xrightarrow[\_]{\vec{\alpha_b'}: (c \mapsto \tuple{\_, B, r.f, u_b'}) \in \vec{\alpha_b'}}$
    $\Sigma_b''$

    between the transitions that contain the actions with ids $u_B$ and
    $u_F$.
  \end{enumerate}

  Note that, as in \wf{10}, \wf{11}, and \wf{14} we do not take in
  account volatile accesses and do not require a program order between
  the actions, instead we require that the corresponding actions are
  performed by the same core $c$.

\item[WF-17] \phantomsection\label{app:wf17}
  Volatile writes are immediately written back.\\

  Allowing other actions between a volatile write and its write-back may
  result in other threads observing these actions as if they were
  executed before the volatile write.
  This is similar to moving these actions before the volatile write,
  which is an invalid reordering according to the JMM\@.
  Formally:

  $ \forall w \in A_D: (w.k = Vw) \Rightarrow \exists b \in A_D: $
  $ \big((w \poD b) \wedge (w.v = b.v) \wedge \nexists x$
  $ \in A_D: (w \poD x \poD b)\big) $\\

  In \lang{} this means that given the execution trace of $E_D$,
  transitions containing volatile write actions:

  $ \sheap; \scaches; \ssdcaches \vdash T: \sproc{c}{r_t, \sassign{r.f}{v}} \in T$
  $ \xrightarrow[\_]{\vec{\alpha}: \tuple{r_t, \mathit{Vw}, r.f, u} \in \rng{\vec{\alpha}}}$
  $ \sheap'; \scaches'; \ssdcaches' \vdash T': \sproc{c}{r_t, v} \in T' $

  update the value of $r.f$ to $v$ in the heap, i.e.:

  \hfil$\left(r \mapsto C(\overrightarrow{f':\tau})\right) \in \sheap' \wedge$
  $ (f \mapsto v) \in (\overrightarrow{f':\tau}) $

\item[WF-18] \phantomsection\label{app:wf18}
  A fetch of the corresponding variable happens
  immediately before each volatile read.\\

  Allowing other actions between a volatile read and its fetch may
  result in other threads observing these actions as if they were
  executed after the volatile read.
  This is similar to moving these actions after the volatile read, which
  is an invalid reordering according to the JMM\@.
  Formally:

  $ \forall r \in A_D: (r.k = Vr) \Rightarrow \exists f \in A_D: $
  $\Big((f \poD r) \wedge \big(f = \mathit{Cs}(r)\big) \wedge$
  $\nexists x \in A_D: (f \poD x \poD r)\Big) $\\

  \noindent{}In \lang{} this means that given the execution trace of
  $E_D$, transitions containing volatile read actions:

  $ \sheap; \scaches; \ssdcaches \vdash T: \sproc{c}{r_t, r.f} \in T$
  $ \xrightarrow[\_]{\vec{\alpha}: \tuple{r_t, \mathit{Vr}, r.f, u} \in \rng{\vec{\alpha}}}$
  $ \sheap'; \scaches'; \ssdcaches' \vdash T': \sproc{c}{r_t, v} \in T' $

  always see the value $v$ of $r.f$ from the heap, i.e.:

  \hfil$\left(r \mapsto C(\overrightarrow{f':\tau})\right) \in \sheap \wedge$
  $ (f \mapsto v) \in (\overrightarrow{f':\tau}) $

\item[WF-19] \phantomsection\label{app:wf19}
  Initializations are immediately written-back and their
  write-backs are completed before the start of any thread.\\

  \noindent{}In \lang{} this rule is always satisfied, since as we
  explain in \wf{9} we define the beginning of every execution trace in
  \lang{} to be $\Sigma_{init} \to^* \Sigma_{init}'$ where $\to^*$
  contains only transitions performing the initialization actions and
  their write-backs, for every variable in the execution trace.
  As a result, in every execution trace initialization actions are
  written-back and their write-backs are completed before the start of
  any thread.

\item[WF-20] \phantomsection\label{app:wf20}
  The happens-before order between two writes is consistent with
  the happens-before order of their write-backs.\\

  If, for two write actions $w$ and $w'$, $w \hbD w'$, then the
  corresponding write-back actions, $b$ for $w$ and $b'$ for $w'$, must
  also be ordered, so that $b \hbD b'$ and vice versa.
  Formally:

  \hfil$
  \forall b,b' \in A_D:  \big(\mathit{Ab}(b) \hb \mathit{Ab}(b')\big)
  \Leftrightarrow (b \hb b')
  $










\item[WFE-1] \phantomsection\label{app:wfe1}
  There is a corresponding fetch action between thread
  migration and every read action.

  $ \forall m,r \in A_D: \big((m.k = M) \wedge (m \poD r)\big)
  \Rightarrow $
  $\big(\exists f \in A_D: (m \poD f \poD r)\big) $\\

  In \lang{}, this means that given the execution trace of $E_D$, for
  every trace:

  $ \ldots \Sigma_1 $
  $\xrightarrow[\_]{\vec{\alpha}: \tuple{r_t, M, \_, u} \in \rng{\vec{\alpha}}} $
  $\Sigma_2 \ldots \Sigma_n $
  $\xrightarrow[\_]{\vec{\alpha'}: \tuple{r_t, R, r.f, u'} \in \rng{\vec{\alpha'}}}$
  $ \Sigma_{n+1} \ldots $

  there exists at least one transition containing a fetch action:

  \hfil$\Sigma $
  $\xrightarrow[\_]{\vec{\alpha}: \tuple{r_t, F, r.f, u_f} \in \rng{\vec{\alpha}}} $
  $\Sigma' $

  between the actions with ids $u$ and $u'$,

  Note that, as in \wf{10}, \wf{11}, \wf{14}, and \wf{16} we do not take
  in account volatile accesses.

\item[WFE-2] \phantomsection\label{app:wfe2} At migration, there are no
  dirty data at the \emph{old} core.
  Formally:

  \hfil$\forall m,w \in A: \Big((m.k = M) \wedge \big(w \po B(w) \po m \big)\Big)$\\

  In \lang{}, this means that given the execution trace of $E_D$, for
  every trace:

  $ \Sigma_1 $
  $ \xrightarrow[\_]{\vec{\alpha_1}: \tuple{r_t, W, r.f, u} \in \rng{\vec{\alpha_1}}} $
  $\Sigma_2 \ldots \Sigma_n $
  $\xrightarrow[\_]{\vec{\alpha_n}: \tuple{r_t, M, \_, u'} \in \rng{\vec{\alpha_n}}}$
  $ \Sigma_{n+1} \ldots $

  there exists at least one transition containing a write-back action
  $\Sigma $
  $\xrightarrow[\_]{\vec{\alpha}: \tuple{r_t, B, r.f, u_f} \in \rng{\vec{\alpha}}} $
  $\Sigma' $ between the actions with ids $u$ and $u'$,
\end{description}

\section{Proof of adherence to JDMM}
\label{sec:proof}

In this section we prove the adherence of \lang{} to JDMM\@.
To achieve this we show that its operational semantics generates only
well-formed, according to JDMM, executions.
That is, given any well-formed execution trace, as described in
\cref{app:jdmm}, $\Sigma \to^* \Sigma'$, where the $\to^*$ binary
operator denotes an arbitrary number of transitions, we show that any
execution trace $\Sigma \to^* \Sigma' \to \Sigma''$ is well-formed as
well.
In our reasoning we introduce some additional well-formedness rules that
we prove true for any \lang{} execution trace.
We mark such rules with \wfh{X}\\

\wfh{1}: For every non-volatile variable $r.f$ that appears in the
execution trace, \emph{if and only if} it is present in $\sheap$, then
its value in $\sheap$ is the one written back by the last, according to
synchronization order, write-back action, acting on $r.f$, in that
execution trace.\\

\wfh{2}: For every non-volatile variable $r.f$ that appears in the
execution trace, \emph{if and only if} it is present in $\sscache(c)$,
then its value in $\sscache(c)$ is the one fetched or written back by
the last fetch or write-back action in that execution trace, which acts
on $r.f$ and is performed by $c$.\\

\wfh{3}: For every non-volatile variable $r.f$ that appears in the
execution trace, \emph{if and only if} it is present in $\ssdcache(c)$,
then its value in $\ssdcache(c)$ is the one written by the last write
action in that execution trace, which acts on $r.f$ and is performed by
$c$.\\

\wfh{4}: For every object $r$ that appears in the execution trace, if
$r \in \dom{\sscache(c)}$, then there is at least one transition
$\Sigma_f $
$\xrightarrow[\vec{c'}: c \in \vec{c'}]{\vec{\alpha}: \tuple{\_, F, r, u_f} \vec{\alpha}} $
$\Sigma_f'$ in the execution trace.\\

\wfh{5}: For every variable $r.f$, that appears in the execution trace,
if:

$r \in \dom{\sheap}$ $\lor$ $r \in \dom{sscache(c)}$ $\lor$
$r.f \in \dom{\ssdcache(c)}$

\noindent{}then the value stored in them is the result of a write to
$r.f$.\\

\wfh{6}: For every volatile variable $r.f$ in $\sheap$, its value is the
one written by the last, according to synchronization order, volatile
write action, acting on it, in that execution trace, or the value
written-back by the write-back action of the initialization action,
acting on it, if there are no volatile write actions, acting on it, in
that execution trace.\\

\wfh{7}: Each thread is assigned to a core \emph{if and only if } it is
spawned, and is assigned to a single core.  Formally,

$\forall c \in \stids: \forall r_t \in \dom{\sheap}:$
$\big(\sheap(r_t) = C(\overrightarrow{f \mapsto v}, \sfmt{spawned})$
$ \iff$
$ \exists T \in \vec{T}: \sproc{c}{r_t, \_} \in T\big) \wedge$

\hfill$\big(\forall T \in \vec{T}: \sproc{c}{r_t, \_} \in T \Rightarrow$
$\nexists c' \in \stids: c' \neq c \wedge \sproc{c'}{r_t, \_} \in T\big)$

where $\vec{T}$ are all the sets of threads in the execution trace.\\

\wfh{8}: Each thread appears only on a single set of threads in a pair
of set of threads.
That is, for every pair of set of threads $T_1 \parallel T_2$ in the
execution trace:

$\forall r_t \in \dom{\sheap}:$
$( \sproc{\_}{r_t, \_} \in T_1 \Rightarrow$
$ \sproc{\_}{r_t, \_} \notin T_2)$
$\wedge$
$( \sproc{\_}{r_t, \_} \in T_2 \Rightarrow$
$ \sproc{\_}{r_t, \_} \notin T_1)$\\

\wfh{9}: The contents of the object cache and the write buffer of each
core are altered only by that core.

$\forall c, c' \in \stids: $
$(\_; \scaches; \ssdcaches \vdash \sproc{c}{\_,\_} \to$
$ \_; \scaches[c' \mapsto \scache{c'}']; \ssdcaches[c' \mapsto \sdcache{c'}'] \vdash \_) \Rightarrow$
$ c = c'$


\begin{lemma}\label{l:wf9}
  Initialization actions happen-before every thread's start action.
\end{lemma}

\begin{proof}
  Satisfied for every execution trace by the definition of the beginning
  of every execution trace in \lang{} to be
  $\Sigma_{init} \to^* \Sigma_{init}'$ where $\to^*$ contains only
  transitions performing the initialization actions and their
  write-backs, for every variable in the execution trace.
\end{proof}

\begin{lemma}[WF-12]\label{l:wf12}
  Fetch actions are preceded by at least one write-back of
  the corresponding variable.
\end{lemma}

\begin{proof}
  In \lang{} this rule is always satisfied, since as we explain in
  \wf{9} we define the beginning of every execution trace in \lang{} to
  be $\Sigma_{init} \to^* \Sigma_{init}'$ where $\to^*$ contains only
  transitions performing the initialization actions and their
  write-backs, for every variable in the execution trace.
\end{proof}

\begin{lemma}[WF-17]\label{l:wf17}
  Volatile writes are immediately written back.
\end{lemma}

\begin{proof}
  Satisfied by the definition of $\rvolatilewrite$ that writes the
  variable directly to the heap.
\end{proof}

\begin{lemma}[WF-18]\label{l:wf18}
  A fetch of the corresponding variable happens immediately before each
  volatile read.
\end{lemma}

\begin{proof}
  Satisfied by the definition of $\rvolatileread$ that reads the
  variable directly from the heap.
\end{proof}

\begin{lemma}[WF-19]\label{l:wf19}
  Initializations are immediately written-back and their
  write-backs are completed before the start of any thread.
\end{lemma}

\begin{proof}
  In \lang{} this rule is always satisfied, since as we explain in
  \wf{9} we define the beginning of every execution trace in \lang{} to
  be $\Sigma_{init} \to^* \Sigma_{init}'$ where $\to^*$ contains only
  transitions performing the initialization actions and their
  write-backs, for every variable in the execution trace.
  As a result, in every execution trace initialization actions are
  written-back and their write-backs are completed before the start of
  any thread.
\end{proof}

\begin{lemma}\label{l:localadherence}
  \lang{}'s local operational semantics generates only well-formed
  execution traces.
\end{lemma}

\begin{proof}
  We show, by induction on the number of steps, that for each well
  formed execution trace $\Sigma \to^* \Sigma'$,
  $\Sigma \to^* \Sigma' \to \Sigma''$, where $\to^*$ and $\to$ are
  reductions of the local operational semantics, is also well-formed.

  Rules $\rctxstep$, $\rcondt$, $\rcondf$, $\rred$, and $\rcall$ regard
  the control flow of the program and are of no interest, since it is
  trivial to show that they preserve the well-formedness of the
  execution.
  Additionally, for each case we omit well-formedness rules that do not
  correlate with the transition at hand, e.g., we do not argue about
  \wf{2} if the rule at hand does not act on a volatile variable.
  Furthermore, we do not argue about \wf{4}, \wf{7} and \wf{8}, since in
  the local operational semantics the happens-before order is equivalent
  to the program order, since the creation of new threads is not
  possible.
  As a result, \wf{4}, \wf{7} and \wf{8} are also satisfied if \wf{6} is
  satisfied.
  Similarly we do not argue about \wf{16} and \wfe{1}-\wfe{2}, since in
  the local operational semantics it is not possible to spawn new
  threads or migrate the main thread, thus all the transitions are
  performed by a single core.\\

  \noindent{}\textbf{Base case:} Any execution trace

  $\Sigma_{init} \to^* $
  $\init \to \Sigma'$

  , is well-formed.\\

  In \lang{} the execution starts with a single thread --the main
  thread-- and the beginning of any execution trace is:

  $\Sigma_{init} \to^* $
  $\init$

  where $\to^*$ contains only transitions performing the initialization
  actions and their write-backs, for every variable in the execution
  trace, and

  $\Sigma_{init} \to^* $
  $\init$

  is well-formed.

  \noindent{}as a result,

  \hfil$\Sigma = \init$\\

  \noindent{}In the local operational semantics,

  \hfil$\sheap; \sscache; \ssdcache \vdash \sproc{c}{e}$ $ \toa{\alpha}$
  $ \sheap; \sscache; \ssdcache \vdash \sproc{c}{e}$

  \noindent{}the only rule that can step is $\rstart$.





  \wf{3} is satisfied, since this is the first synchronization action,
  other than initialization actions, in the execution trace and the
  number of initialization actions is equal to the number of variables,
  in a program, which is finite.\\

  \wf{9} is satisfied by \cref{l:wf9} and the fact that the action at
  hand is a start action and is the first action, other than
  initialization and write-backs, in the program.\\




  \wfh{7} is satisfied, since initially there only exists a single
  thread, the main thread, that starts in a single core.\\

  The rest of the rules are omitted since they do not correlate with the
  transition at hand and thus it is trivial to show that they are
  satisfied.

  As a result, the lemma is true for the first transition of any program.

  ~\\
  \textbf{Inductive step:} Given a well-formed execution trace
  $\Sigma \to^* \Sigma'$, $\Sigma \to^* \Sigma' \to \Sigma''$ is also
  well-formed.

  We examine each case for $\Sigma' \to \Sigma''$ in the local
  operational semantics:

  \hfil$\sheap; \sscache; \ssdcache \vdash \sproc{c}{r_t, e}$
  $ \toa{\alpha}$
  $ \sheap; \sscache; \ssdcache \vdash \sproc{c}{r_t, e}$

  \noindent{}and show that it satisfies the well-formedness rules.

  \begin{case}
    $\rfield$
  \end{case}

  \noindent\hfil$\Sigma \to^* \Sigma'$
  $ \toa{\tuple{r_t, R, r.f, u}}$
  $ \sheap; \sscache; \ssdcache \vdash \sproc{c}{r_t, e'}$\\

  \noindent{}where $r.f \notin \dom{\ssdcache}$ and
  $\Sigma' = \sheap; \sscache; \ssdcache \vdash \sproc{c}{r_t, e}$.\\

  By the premises of $\rfield$:

  \hfil$ r \in \dom{\sheap} \wedge \neg \vol{v.f} \wedge \sscache(r.f) = v $\\

  \wf{1}: Since the value of $r.f$ is read from the object cache and
  $\Sigma \to^* \Sigma'$ is well formed, according to \wfh{5} that value
  will be the result of a write action, acting on $r.f$, that is
  performed by a transition in the execution trace.
  As a result, \wf{1} is satisfied.\\

  \wf{6}: Since $r.f$ is present in the object cache and
  $\Sigma \to^* \Sigma'$ is well formed and according to \wfh{2}, it was
  either fetched or updated through a write-back.
  In both cases, since $\Sigma \to^* \Sigma'$ is well formed, according
  to \wfh{1} and \wfh{2}, respectively, the cached value will be that of
  the last write-back in the execution trace.
  Additionally, according to \wf{20} the happens-before order between
  two writes is consistent with the happens-before order of their
  write-backs, meaning that the cached value will be that of the last
  write in the execution trace.
  That said, \wf{6} is satisfied.\\

  \wf{10}: Since $\Sigma \to^* \Sigma'$ is well formed and
  $\sscache(r.f) = v$, according to \wfh{4}, there exists a transition
  $\Sigma_f \toa{\tuple{\_, F, r, u_f}} \Sigma_f'$ in
  $\Sigma \to^* \Sigma'$.
  As a result, \wf{10} is also satisfied.\\

  \wf{11}: Since the value of $r.f$ is read from the object cache and
  $\Sigma \to^* \Sigma'$ is well formed, according to \wfh{2} that value
  will be the result of the last fetch or write-back action, acting on
  $r.f$, that is performed by a transition in the execution trace.
  As a result, there are no updates or overwrites of the cached value
  between between the value that cached it and the read that sees it.
  An invalidation of $r.f$ between the last, in the execution trace,
  fetch or write-back action, that cached $r.f$, and the read, would
  result in the premises of $\rfield$ not being satisfied, since the
  object cache would not contain a value for $r.f$.
  As a result there is also no invalidation of the variable's cached
  value between the action that cached it and the read that sees it.
  As a result, \wf{11} is satisfied.\\

  The rest of the rules are omitted since they do not correlate with the
  transition at hand and thus it is trivial to show that they are
  satisfied.

  \begin{case}
    $\rfieldd$
  \end{case}

  \noindent\hfil$\Sigma \to^* \sheap; \sscache; \ssdcache \vdash \sproc{c}{r_t, e}$
  $ \toa{\tuple{r_t, R, r.f, u}}$
  $ \sheap; \sscache; \ssdcache \vdash \sproc{c}{r_t, e'}$

  \noindent{}where $r.f \in \dom{\ssdcache}$.\\

  By the premises of $\rfieldd$:

  \hfil$ r \in \dom{\sheap} \wedge \neg \vol{v.f} \wedge \sscache(r.f) = v' $

  wf{1}: Since the value of $r.f$ is read from the write buffer and
  $\Sigma \to^* \Sigma'$ is well formed, according to \wfh{5} that value
  will be the result of a write action, acting on $r.f$, performed by a
  transition in the execution trace.
  As a result, \wf{1} is satisfied.\\

  \wf{6}: Since the value of $r.f$ is read from the write buffer and
  $\Sigma \to^* \Sigma'$ is well formed, according to \wfh{3} that value
  will be the result of the last write action, acting on $r.f$, that is
  performed by a transition in the execution trace.
  As a result, \wf{6} is satisfied.\\

  \wf{11}: Since the value is read from the write buffer and
  $\Sigma \to^* \Sigma'$ is well formed, according to \wfh{3} that value
  will be the result of the last write action, acting on $r.f$, that is
  performed by a transition in the execution trace.
  As a result, there are no updates or overwrites of the cached value
  between between the value that cached it and the read that sees it.
  Additionally, an invalidation of $r.f$ (possible through
  $\rwriteback$) between the last, in the execution trace, write action
  that added $r.f$ to the write buffer and the read would result in the
  premises of $\rfieldd$ not being satisfied, since the write buffer
  would not contain a value for $r.f$.
  As a result there is also no invalidation of the variable's cached
  value between the action that cached it and the read that sees it.
  As a result, \wf{11} is satisfied.\\

  The rest of the rules are omitted since they do not correlate with the
  transition at hand and thus it is trivial to show that they are
  satisfied.

  \begin{case}
    $\rassign$
  \end{case}

  \noindent\hfil$\Sigma \to^* \sheap; \sscache; \ssdcache \vdash \sproc{c}{r_t, e}$
  $ \toa{\tuple{r_t, W, r.f, u}}$
  $ \sheap; \sscache; \ssdcache' \vdash \sproc{c}{r_t, e'}$

  \noindent{}where $\ssdcache' = \ssdcache[r.f \mapsto v]$.\\

  By the premises of $\rassign$:

  \hfil$ r \in \dom{\sheap} \wedge \neg \vol{v.f}$

  \wfh{3} and \wfh{5} are satisfied since the new value of $r.f$ in the
  write buffer is the one written by the write action of the last
  transition in the execution trace.\\

  \wfh{9} is satisfied, since the new value is added to the write buffer
  of the core performing the action.\\

  The rest of the rules are omitted since they do not correlate with the
  transition at hand and thus it is trivial to show that they are
  satisfied.

  \begin{case}
    $\rnew$
  \end{case}

  \noindent\hfil$\Sigma \to^* \sheap; \sscache; \ssdcache \vdash \sproc{c}{r_t, e}$
  $ \toa{\tuple{r_t, \mathit{In}, r.f, u}}$
  $ \sheap; \sscache; \ssdcache \vdash \sproc{c}{r_t, e'}$\\

  \noindent{}where
  $r - \mathit{fresh} \wedge$
  $ \sheap' = \sheap[r \mapsto C(\overrightarrow{f \mapsto 0})] \wedge$
  $ C(\overrightarrow{f:\tau})\{e\} \in C$\\

  \wfh{1}, \wfh{5}, and \wfh{6} are satisfied since the values of the
  new object's variables in the heap are those of the last write-back to
  these variables, namely the write-back of their initialization.\\

  \wfh{2}--\wfh{3} are satisfied since they are satisfied in
  $\Sigma \to^* \Sigma'$ and $\rnew$ does not modify the object
  cache, or the write buffer.\\

  The rest of the rules are omitted since they do not correlate with the
  transition at hand and thus it is trivial to show that they are
  satisfied.

  \begin{case}
    $\rvolatilereadl$
  \end{case}

  \wf{5} is satisfied since it is satisfied in
  $\Sigma \to^* \Sigma'$
  and $\rvolatilereadl$ requires $r.f.l$ to be free before acquiring it.\\

  \wfh{1}, \wfh{5}, and \wfh{6} are satisfied since it is satisfied in
  $\Sigma \to^* \Sigma'$ and $\rvolatilereadl$ does not modify any
  variables in the heap, only the synthetic lock of the volatile
  variable at hand.\\

  The rest of the rules are omitted since they do not correlate with the
  transition at hand and thus it is trivial to show that they are
  satisfied.

  \begin{case}
    $\rvolatileread$
  \end{case}

  \noindent\hfil $\Sigma \to^* \sheap; \sscache; \ssdcache \vdash \sproc{c}{r_t, e}$
  $ \toa{\tuple{r_t, \mathit{Vr}, r.f, u}} $
  $ \sheap'; \sscache; \ssdcache \vdash \sproc{c}{r_t, e'}$\\

  \noindent{}where
  $ r \in \dom{\sheap} \wedge $
  $ \sheap (r.f.l) = r_t \wedge $ $\sscache = \emptyset \wedge $
  $ \ssdcache = \emptyset \wedge $
  $ \sheap' = \sheap[r.f.l \mapsto 0] \wedge $
  $ \sheap(r.f) = v $\\

  wf{1}: Since the value of $r.f$ is read from the heap and
  $\Sigma \to^* \Sigma'$ is well formed, according to \wfh{6} that value
  will be the result of the last volatile write action, acting on $r.f$,
  in that execution trace, or by the initialization action, acting on
  $r.f$, if there are no volatile write actions, acting on $r.f$, in
  that execution trace.
  As a result, \wf{1} and \wf{6} are satisfied.\\

  \wf{2} is satisfied since in $\Sigma \to^* \Sigma'$ all volatile
  variables where accessed by volatile actions according to \wf{2} and
  the volatile read at hand is also a volatile action.\\

  \wf{3} is satisfied, since $\Sigma \to^* \Sigma'$ is well formed and
  according to \wf{3} the number of synchronization actions in it are
  finite.\\

  \wfh{1}, \wfh{5}, and \wfh{6} are satisfied since it is satisfied in
  $\Sigma \to^* \Sigma'$ and $\rvolatilewritel$ does not modify any
  variables in the heap, only the synthetic lock of the volatile
  variable at hand.\\

  The rest of the rules are omitted since they do not correlate with the
  transition at hand and thus it is trivial to show that they are
  satisfied.

  \begin{case}
    $\rvolatilewritel$
  \end{case}

  \wf{5} is satisfied since it is satisfied in
  $\Sigma \to^* \Sigma'$
  and $\rvolatilewritel$ requires $r.f.l$ to be free before acquiring
  it.\\

  \wfh{1}, \wfh{5}, and \wfh{6} are satisfied since it is satisfied in
  $\Sigma \to^* \Sigma'$ and $\rvolatilewritel$ does not modify any
  variables in the heap, only the synthetic lock of the volatile
  variable at hand.\\

  The rest of the rules are omitted since they do not correlate with the
  transition at hand and thus it is trivial to show that they are
  satisfied.

  \begin{case}
    $\rvolatilewrite$
  \end{case}

  \wf{2} is satisfied since in $\Sigma \to^* \Sigma'$ all volatile
  variables where accessed by volatile actions according to \wf{2} and
  the volatile write at hand is also a volatile action.\\

  \wf{3} is satisfied, since $\Sigma \to^* \Sigma'$ is well formed and
  according to \wf{3} the number of synchronization actions in it are
  finite.\\

  \wf{5} is satisfied since it is satisfied in
  $\Sigma \to^* \Sigma'$
  and $\rvolatilewrite$ requires $r.f.l$ to be acquired by the thread
  performing the action to release it.\\

  \wfh{1} is satisfied since it is satisfied in $\Sigma \to^* \Sigma'$
  and $\rvolatilewrite$ does not modify any non-volatile variables in
  the heap.\\

  \wfh{5} and \wfh{6} are satisfied since the new value of $r.f$ in the
  heap is the one written by the volatile write action of the last
  transition in the execution trace.\\

  The rest of the rules are omitted since they do not correlate with the
  transition at hand and thus it is trivial to show that they are
  satisfied.

  \begin{case}
    $\rmonitorentera$
  \end{case}

  \wf{3} is satisfied, since $\Sigma \to^* \Sigma'$ is well formed and
  according to \wf{3} the number of synchronization actions in it are
  finite.\\

  \wf{5} is satisfied since it is satisfied in $\Sigma \to^* \Sigma'$
  and $\rmonitorentera$ requires that the monitor $r.l$ is free before
  acquiring it.\\

  \wfh{1}, \wfh{5}, and \wfh{6} are satisfied since it is satisfied in
  $\Sigma \to^* \Sigma'$ and $\rvolatilewritel$ does not modify any
  variables in the heap, only the monitor of the object at hand.\\

  The rest of the rules are omitted since they do not correlate with the
  transition at hand and thus it is trivial to show that they are
  satisfied.

  \begin{case}
    $\rmonitorenterb$
  \end{case}

  \wf{3} is satisfied, since $\Sigma \to^* \Sigma'$ is well formed and
  according to \wf{3} the number of synchronization actions in it are
  finite.\\

  \wf{5} is satisfied since it is satisfied in
  $\Sigma \to^* \Sigma'$
  and $\rmonitorenterb$ requires that the monitor $r.l$ is already
  acquired by the thread performing the action in order to re-acquire
  it.\\

  \wfh{1}, \wfh{5}, and \wfh{6} are satisfied since it is satisfied in
  $\Sigma \to^* \Sigma'$ and $\rvolatilewritel$ does not modify any
  variables in the heap, only the monitor of the object at hand.\\

  The rest of the rules are omitted since they do not correlate with the
  transition at hand and thus it is trivial to show that they are
  satisfied.

  \begin{case}
    $\rmonitorexita$
  \end{case}

  \wf{3} is satisfied, since $\Sigma \to^* \Sigma'$ is well formed and
  according to \wf{3} the number of synchronization actions in it are
  finite.\\

  \wf{5} is satisfied since it is satisfied in $\Sigma \to^* \Sigma'$
  and $\rmonitorexita$ requires that the monitor $r.l$ is already
  acquired a single time by the thread performing the action in order to
  release it.\\

  \wfh{1}, \wfh{5}, and \wfh{6} are satisfied since it is satisfied in
  $\Sigma \to^* \Sigma'$ and $\rvolatilewritel$ does not modify any
  variables in the heap, only the monitor of the object at hand.\\

  The rest of the rules are omitted since they do not correlate with the
  transition at hand and thus it is trivial to show that they are
  satisfied.

  \begin{case}
    $\rmonitorexitb$
  \end{case}

  \wf{3} is satisfied, since $\Sigma \to^* \Sigma'$ is well formed and
  according to \wf{3} the number of synchronization actions in it are
  finite.\\

  \wf{5} is satisfied, since it is satisfied in $\Sigma \to^* \Sigma'$
  and $\rmonitorexitb$ requires that the monitor $r.l$ is already
  acquired more than one times by the thread performing the action in
  order to decrease by one the acquisitions by that thread.\\

  \wfh{1}, \wfh{5}, and \wfh{6} are satisfied since it is satisfied in
  $\Sigma \to^* \Sigma'$ and $\rvolatilewritel$ does not modify any
  variables in the heap, only the monitor of the object at hand.\\

  The rest of the rules are omitted since they do not correlate with the
  transition at hand and thus it is trivial to show that they are
  satisfied.

  \begin{case}
    $\racquire$
  \end{case}

  \wf{3} is satisfied, since $\Sigma \to^* \Sigma'$ is well formed and
  according to \wf{3} the number of synchronization actions in it are
  finite.\\

  The rest of the rules are omitted since they do not correlate with the
  transition at hand and thus it is trivial to show that they are
  satisfied.

  \begin{case}
    $\rrelease$
  \end{case}

  \wf{3} is satisfied, since $\Sigma \to^* \Sigma'$ is well formed and
  according to \wf{3} the number of synchronization actions in it are
  finite.\\

  The rest of the rules are omitted since they do not correlate with the
  transition at hand and thus it is trivial to show that they are
  satisfied.

  \begin{case}
    $\rfetch$
  \end{case}

  \wf{1}: Since $r$ and its variables are fetched from the heap and
  $\Sigma \to^* \Sigma'$ is well formed, according to \wfh{1} for each
  variable $r.f$ in $r$ its value is the one written back by the last
  write-back action, acting on $r.f$, in that execution trace.
  As a result, \wf{12} is satisfied.\\

  \wf{3} is satisfied, since $\Sigma \to^* \Sigma'$ is well formed and
  according to \wf{3} the number of synchronization actions in it is
  finite.\\

  \wfh{2} and \wfh{4} are satisfied since the value of $r.f$ in the
  object cache is the one fetched from the last fetch action in the
  execution trace.\\

  \wfh{9} is satisfied, since the fetch value is added to the object
  cache of the core performing the action.\\

  The rest of the rules are omitted since they do not correlate with the
  transition at hand and thus it is trivial to show that they are
  satisfied.

  \begin{case}
    $\rwriteback$
  \end{case}

  \wf{3} is satisfied, since $\Sigma \to^* \Sigma'$ is well formed and
  according to \wf{3} the number of synchronization actions in it is
  finite.\\

  \wf{13} and \wf{14}: Since $r.f$ is written back from the write buffer
  and $\Sigma \to^* \Sigma'$ is well formed, according to \wfh{3} its
  value in the write buffer is the one written by the last write action
  in that execution trace, which acts on $r.f$ and is performed by $c$.
  As a result, \wf{13} and \wf{14} are satisfied.\\

  \wf{20}: Since $\Sigma \to^* \Sigma'$ is well-formed \wf{20} is
  satisfied for any pair of writes and the corresponding pair of their
  write-backs in it.
  As a result, we examine the cases where the second write $w$ of the
  pair is the last write in the trace, which the write-back $b$ at hand
  writes back.
  Given any pair of write and write-back actions $w'$ and $b'$ in
  $\Sigma \to^* \Sigma'$ (if there exists one), where $w' \hbD w$,
  according to \wf{14} the write-back action $b'$ writing back $w'$ can
  only appear between the two writes $w' \hbD b' \hbD w$.
  Additionally, we know that $w \poD b$.
  As a result, $w' \hbD b' \hbD w \hbD b$ which satisfies \wf{20}.\\

  \wfh{1} is satisfied since the value of $r.f$ in the heap is the one
  written back by the last write-back action in the execution trace.\\

  \wfh{2} is satisfied since the value of $r.f$ in the object cache is
  the one written back by the last write-back action in the execution
  trace.\\

  \wfh{3} and \wfh{5} are satisfied since $\rwriteback$ just removes
  $r.f$ from the write buffer and does not introduce or restore another
  value in its place.\\

  \wfh{9} is satisfied, since the value is moved from the write buffer,
  of the core performing the action, to the object cache of the same
  core.\\


  The rest of the rules are omitted since they do not correlate with the
  transition at hand and thus it is trivial to show that they are
  satisfied.

  \begin{case}
    $\rinvalidate$
  \end{case}

  \wf{3} is satisfied, since $\Sigma \to^* \Sigma'$ is well formed and
  according to \wf{3} the number of synchronization actions in it is
  finite.\\

  \wf{15} is satisfied, since the first premise of $\rinvalidate$
  requires that the object being invalidated is present in the object
  cache.
  As a result, only cached variables are invalidated.\\

  \wfh{2} and \wfh{5} are satisfied since $\rinvalidate$ just removes a
  value from the object cache and does not introduce or restore another
  value in its place.\\

  \wfh{9} is satisfied, since the value is removed from the object cache
  of the core performing the action.\\

  The rest of the rules are omitted since they do not correlate with the
  transition at hand and thus it is trivial to show that they are
  satisfied.

  \begin{case}
    $\rstart$
  \end{case}

  \wf{3} is satisfied, since $\Sigma \to^* \Sigma'$ is well formed and
  according to \wf{3} the number of synchronization actions in it are
  finite.\\

  \wf{9} is satisfied by \cref{l:wf9} and the fact that in the local
  operational semantics there is no way to step to the start
  expression.  The only start exception in
  the program is that of the main thread in the initial state.\\

  The rest of the rules are omitted since they do not correlate with the
  transition at hand and thus it is trivial to show that they are
  satisfied.\\

  \noindent{}\lang{}'s local operational semantics generates only well-formed
  execution traces.
\end{proof}


\begin{lemma}\label{l:lift}
  Lifting a well-formed execution trace from the local operational
  semantics to the global operational semantics preserves the
  well-formedness of the execution.
\end{lemma}

\begin{proof}
  In \lang{} the lifting is performed by $\rlift$.
  $\rlift$ does not introduce new modifications to the memory state or
  new actions in the execution trace, other than those performed by the
  local operational semantics.
  As a result, since according to \cref{l:localadherence} the local
  operational semantics only generates well formed executions, lifting
  it to the global operational semantics preserves its well-formedness.
\end{proof}

\begin{theorem}\label{t:adherence}
  \lang{}'s operational semantics generates only well-formed execution
  traces.  As a result, all executions performed by DiSquawk adhere to
  JDMM and consequently to JMM.
\end{theorem}

\begin{proof}
  We show, by induction on the number of steps, that for each well
  formed execution trace $\Sigma \to^* \Sigma'$,
  $\Sigma \to^* \Sigma' \to \Sigma''$, where $\to^*$ and $\to$ are
  reductions of the global operational semantics, is also well-formed.

  For each case we omit well-formedness rules that do not correlate with
  the transition at hand, e.g., we do not argue about \wf{2} in the case
  of $\rspawn$ since it does not act on a volatile variable.\\

  \textbf{Base case:} Any execution trace $\Sigma \to \Sigma'$, is
  well-formed.\\

  In \lang{} the execution starts with a single thread --the main
  thread-- and the beginning of any execution trace is:

  $\Sigma_{init} \to^* \init$

  \noindent{}where $\to^*$ contains only transitions performing the
  initialization actions and their write-backs, for every variable in
  the execution trace, and $\Sigma_{init} \to^* \init$ is well-formed.

  \noindent{}As a result, $\Sigma = \init$\\

  \noindent{}In the global operational semantics,

  \hfil$\sheap; \scaches; \ssdcaches \vdash T$
  $ \xrightarrow[\vec{c}]{\vec{\alpha}}$
  $ \sheap; \scaches; \ssdcaches \vdash T $

  \noindent{}the interesting cases are $\rlift$ and $\rmigrate$.
  $\rspawn$ cannot step since its premises are not satisfied.
  $\rpar$ does not change the state and for $\rparg$ there is no other
  thread in the context to step.

  \begin{tcase}
    $\rlift$
  \end{tcase}

  In the case of $\rlift$, the well-formedness of the execution is
  preserved according to \cref{l:lift}.

  \begin{tcase}
    $\rmigrate$
  \end{tcase}

  In the case of $\rmigrate$ the main thread is transferred to another
  core.
  The memory state remains as before and all well-formedness rules are
  satisfied.

  \wfe{2} is satisfied by $\rmigrate$'s premises ---there are no data in
  the write buffer.

  \wfh{7}: Since $\Sigma \to^* \Sigma'$ is well formed and satisfies
  \wfh{7}, the thread at hand is spawned.
  $\rmigrate$ transfers the thread at hand to a new core and resigns it
  from its previous core complying to \wfh{7}.
  As a result, \wfh{7} is satisfied.

  The rest of the rules are omitted since they do not correlate with the
  transition at hand and thus it is trivial to show that they are
  satisfied.

  As a result, the theorem is true for the first transition of any
  program.

  ~\\
  \textbf{Inductive step:} Given a well-formed execution trace
  $\Sigma \to^* \Sigma'$, $\Sigma \to^* \Sigma' \to \Sigma''$ is also
  well-formed.

  We examine each case in the global operational semantics:

  \hfil$\sheap; \scaches; \ssdcaches \vdash T$
  $ \xrightarrow[\vec{c}]{\vec{\alpha}}$
  $ \sheap; \scaches; \ssdcaches \vdash T $

  \noindent{}and show that it satisfies the well-formedness rules.

  \begin{tcase}
    $\rlift$
  \end{tcase}

  In the case of $\rlift$, the well-formedness of the execution is
  preserved according to \cref{l:lift}.

  \begin{tcase}
    $\rspawn$
  \end{tcase}

  \wf{4}: Since $\Sigma \to^* \Sigma'$ is well formed, according to
  \wf{4}, synchronization order is consistent with program order.
  The action at hand is placed after, according to the program order and
  the synchronization order, any actions in $\Sigma \to^* \Sigma'$.
  As a result the synchronization order remains consistent with the
  program order and \wf{4} is satisfied.

  \wfh{7} and \wfh{8}: The spawned thread is assigned to a single core
  and the old thread remains assigned to its core.
  The spawned thread also gets marked as spawned in order to forbid
  future re-spawns of the same thread (first and second premise of
  $\rspawn$).
  As a result, \wfh{7} and \wfh{8} are satisfied, since they are also
  satisfied in $\Sigma \to^* \Sigma'$.

  The rest of the rules are omitted since they do not correlate with the
  transition at hand and thus it is trivial to show that they are
  satisfied.

  \begin{tcase}
    $\rmigrate$
  \end{tcase}

  \wfe{2} is satisfied since it is satisfied in $\Sigma \to^* \Sigma'$
  and in the new transition is satisfied by $\rmigrate$'s premises
  ---there are no data in the write buffer.

  \wfh{7}: Since $\Sigma \to^* \Sigma'$ is well formed and satisfies
  \wfh{7}, the thread at hand is spawned.
  $\rmigrate$ transfers the thread at hand to a new core and resigns it
  from its previous core complying to \wfh{7}.
  As a result, \wfh{7} is satisfied.

  The rest of the rules are omitted since they do not correlate with the
  transition at hand and thus it is trivial to show that they are
  satisfied.

  \begin{tcase}
    $\rpar$
  \end{tcase}

  In the case of $\rpar$ all well formed rules are satisfied since they
  where satisfied in $\Sigma \to^* \Sigma'$ and $\rpar$ does not
  introduce any state modifications or new actions in the execution
  trace.

  \begin{tcase}
    $\rparg$
  \end{tcase}

  \wf{1}: Since $\Sigma \to^* \Sigma'$ is well-formed, \wf{1} and
  \wfh{5} are true for it.

  In the case of non-volatile reads the read of a variable $r.f$ sees
  the value written in the object cache or the write buffer of the core
  that performs the action (see $\rfield$ and $\rfieldd$), which
  according to \wfh{5} is the result of a write to $r.f$.
  Since the object caches and the write buffers of different cores are
  disjoint \wf{1} and \wfh{5} are true for the unions of the object
  caches and the write buffers as well.

  In the case of volatile reads the read of a volatile variable $r.f$
  sees the value written in the heap (see $\rvolatileread$), which
  according to \wfh{5} is the result of a write to $r.f$.
  By induction on the eighth premise of $\rparg$, only one core may
  modify the heap.
  Since $\rvolatileread$ modifies it, then there are no writes to the
  heap executed in parallel with $\rvolatileread$ and the latter will
  see the last write to $r.f$, according to \wfh{6}, since
  $\Sigma \to^* \Sigma'$ is well-formed.
  As a result, \wf{1} is satisfied.\\

  \wf{2}: By induction on the eighth and ninth premise of $\rparg$,
  every thread steps through the $\rlift$, $\rspawn$, or $\rmigrate$.
  According to \cref{l:lift} every step performed by $\rlift$ is well
  formed and thus satisfies \wf{2}.
  $\rspawn$ and $\rmigrate$ do not act on volatile variables, so they
  always preserve \wf{2}.
  As a result \wf{2} is satisfied.\\

  \wf{3}: Since $\Sigma \to^* \Sigma'$ is well-formed, \wf{3}, \wfh{7},
  and \wfh{8} are true for it, as a consequence, the number of spawned
  threads in the system is finite, since the spawn action is a
  synchronization action.
  Additionally by \wfh{7} each spawned thread is assigned to a single
  core and by \wfh{8} each thread appears only on a single set of
  threads.
  As a result, the number of synchronization actions that can be
  performed in parallel is bound by the number of the spawned threads in
  the system.
  As a result, \wf{3} is satisfied.\\

  \wf{4}: Since $\Sigma \to^* \Sigma'$ is well-formed, \wf{4}, \wfh{7},
  and \wfh{8} are true for it.
  By \wfh{7} each spawned thread is assigned to a single core and by
  \wfh{8} each thread appears only on a single set of threads.
  As a result, a thread may not step in parallel with itself, and any
  action is appended to the program order.
  However, in the case of synchronization actions, $F$, $I$,
  $\mathit{J}$, and $\mathit{Ird}$ may step in parallel with other
  synchronization actions, so they are not actually ordered with those
  actions.
  Nevertheless, any arbitrary ordering of them does not break the
  consistency of the synchronization order with the program order, since
  only a single action maybe performed by each thread in every
  transition.
  As a result, \wf{4} is satisfied.\\

  \wf{5}: Since $\Sigma \to^* \Sigma'$ is well-formed, \wf{5} is true for
  it.
  Additionally, only a single lock operation may be performed at any
  parallel transition, since lock operations modify the heap and
  according to the eighth and ninth premises of $\rparg$ only one set of
  threads is allowed to modify it.
  By induction on the eighth premise we conclude that only a single
  thread may modify the heap, through $\rlift$.
  Since according to \cref{l:lift} $\rlift$ preserves the
  well-formedness, \wf{5} is satisfied by $\rparg$ as well.\\

  \wf{6}: Since $\Sigma \to^* \Sigma'$ is well-formed, \wf{6}, \wfh{7},
  and \wfh{8} are true for it.
  By \wfh{7} each spawned thread is assigned to a single core and by
  \wfh{8} each thread appears only on a single set of threads.
  As a result, a thread may not step in parallel with itself, and any
  action is appended to the program order.
  By induction on the eighth and ninth premise of $\rparg$, every thread
  steps through the $\rlift$, $\rspawn$, or $\rmigrate$.
  According to \cref{l:lift} every step performed by $\rlift$ is well
  formed and thus satisfies \wf{6}.
  $\rspawn$ and $\rmigrate$ do not perform any reads, so they always
  satisfy \wf{6}.\\

  \wf{7}: By induction on the eighth and ninth premise of $\rparg$,
  every thread steps through the $\rlift$, $\rspawn$, or $\rmigrate$.
  According to \cref{l:lift} every step performed by $\rlift$ is well
  formed and thus satisfies \wf{7}.
  $\rspawn$ and $\rmigrate$ do not correspond to volatile actions, so
  they always preserve \wf{7}.
  Additionally, by induction on the eighth premise of $\rparg$, only one
  core may modify the heap.
  Since volatile actions modify it, then there are no other volatile
  actions executed in parallel with $\rvolatileread$ and the latter will
  see the last write to $r.f$, according to \wfh{6}, since
  $\Sigma \to^* \Sigma'$ is well-formed.
  As a result, \wf{7} is satisfied.\\

  \wf{8}: The happens-before order is the transitive closure of the
  synchronizes-with order and the program order.

  As we show for \wf{6}, since $\Sigma \to^* \Sigma'$ is well-formed,
  \wf{6}, \wfh{7}, and \wfh{8} are true for it.
  By \wfh{7} each spawned thread is assigned to a single core and by
  \wfh{8} each thread appears only on a single set of threads.
  As a result, a thread may not step in parallel with itself, and any
  action is appended to the program order.

  Regarding the synchronizes-with order, we examine each pair and show
  that both actions of a pair can not step in parallel.
  Note that we omit the last pair regarding finalization and the
  constructor of the object, since we do not model finalization in our
  semantics.

  \begin{itemize}
  \item $\mathit{In} \swD S$: According to \cref{l:wf9} initialization
    actions are performed before the start of the program.
  \item $\mathit{Vw} \swD \mathit{Vr}$: Since both $\mathit{Vw}$ and
    $\mathit{Vr}$ modify the heap they cannot step in parallel.
    By induction on the eighth premise of $\rparg$, only one core may
    modify the heap.
  \item $U \swD L$: Since both $U$ and $L$ modify the heap they cannot
    step in parallel.
    By induction on the eighth premise of $\rparg$, only one core may
    modify the heap.
  \item $\mathit{Sp} \swD S$: Since both $\mathit{Sp}$ and $S$ modify
    the heap they cannot step in parallel.
    By induction on the eighth premise of $\rparg$, only one core may
    modify the heap.
  \item $\mathit{Fi} \swD J$: Since $\mathit{Fi}$ modifies the heap and
    $J$ reads it, although they are allowed to step in parallel by
    $\rparg$, the third premise of $\rjoin$ would not be satisfied, as a
    result they never step in parallel.
  \item $\mathit{Ir} \swD \mathit{Ird}$: Since $\mathit{Ir}$ modifies
    the heap and $\mathit{Ird}$ reads it, although they are allowed to
    step in parallel by $\rparg$, the third premise of $\rinterruptedt$
    would not be satisfied, as a result they never step in parallel.
  \end{itemize}

  As a result, \wf{8} is satisfied.\\

  \wf{9}: According to \cref{l:wf9} every initialization action in the
  execution trace happens-before the start of the program.
  Additionally, since $\Sigma \to^* \Sigma'$ is well-formed, \wf{9}
  is true for it and start actions modify the heap to mark the thread as
  started.
  By induction on the eighth premise of $\rparg$, only one core may
  modify the heap.
  As a result, there can only be a single start action in a parallel
  transition and that will be evaluated by $\rlift$ that according to
  \cref{l:lift} preserves the well-formedness of the execution.
  That is, in the execution trace preceding the transition at hand all
  thread actions where ordered after the start action of the
  corresponding thread according to the happens-before order.
  Additionally, the same is true for the local execution trace of the core
  that starts the thread.
  As a result the only case that remains to be examined is that of
  running a start action in parallel with another action of that thread.
  Since $\Sigma \to^* \Sigma'$ is well-formed, \wf{6}, \wfh{7}, and
  \wfh{8} are true for it.
  By \wfh{7} each spawned thread is assigned to a single core and by
  \wfh{8} each thread appears only on a single set of threads.
  As a result, a thread may not step in parallel with itself, and any
  action is appended to the program order.
  As a result \wf{9} is satisfied.\\

  \wf{10}: By induction on the eighth and ninth premise of $\rparg$,
  every thread steps through the $\rlift$, $\rspawn$, or $\rmigrate$,
  and specifically reads step through $\rlift$.
  According to \cref{l:lift} every step performed by $\rlift$ is well
  formed and thus satisfies \wf{10}.
  As a result, there is a write or fetch action, acting on the same
  variable as the read, earlier in the execution trace.\\

  \wf{11}: By induction on the eighth and ninth premise of $\rparg$,
  every thread steps through the $\rlift$, $\rspawn$, or $\rmigrate$,
  and specifically reads step through $\rlift$.
  According to \cref{l:lift} every step performed by $\rlift$ is well
  formed and thus satisfies \wf{11}.
  As a result for each non-volatile read there is no invalidation,
  update, or overwrite of the variable's value between the read and
  fetch or write that cached it.
  By \wfh{7} on $\Sigma \to^* \Sigma'$ each spawned thread is assigned
  to a single core, by \wfh{8} each thread appears only on a single set
  of threads, and by \wfh{9} the contents of the object cache and the
  write buffer of each core are altered only by that core.
  As a result, since \wf{9} holds by $\rlift$ it is also true for the
  whole transition, since the core performing the read is the only that
  can alter the object cache and the write buffer, and it cannot perform
  another action in parallel with itself (first premise of
  $\rparg$), to invalidate, update, or overwrite the value.\\

  \wf{12}: See \cref{l:wf12}.

  \wf{13}: By induction on the eighth and ninth premise of $\rparg$,
  every thread steps through the $\rlift$, $\rspawn$, or $\rmigrate$,
  and specifically write-backs step through $\rlift$.
  According to \cref{l:lift} every step performed by $\rlift$ is well
  formed and thus satisfies \wf{13}.
  As a result, there is a write, to the corresponding variable
  being written-back, earlier in the execution trace.\\

  \wf{14}: By induction on the eighth and ninth premise of $\rparg$,
  every thread steps through the $\rlift$, $\rspawn$, or $\rmigrate$,
  and specifically write-backs step through $\rlift$.
  According to \cref{l:lift} every step performed by $\rlift$ is well
  formed and thus satisfies \wf{14}.
  By \wfh{14} on $\Sigma \to^* \Sigma'$ each spawned thread is assigned
  to a single core, by \wfh{8} each thread appears only on a single set
  of threads, and by \wfh{9} the contents of the object cache and the
  write buffer of each core are altered only by that core.
  As a result, since \wf{14} holds by $\rlift$ it is also true for the
  whole transition, since the core performing the write-back is the only
  that can alter the object cache and the write buffer and it cannot
  perform a write action in parallel with itself (first premise of
  $\rparg$).\\

  \wf{15}: By induction on the eighth and ninth premise of $\rparg$,
  every thread steps through the $\rlift$, $\rspawn$, or $\rmigrate$,
  and specifically invalidations step through $\rlift$.
  According to \cref{l:lift} every step performed by $\rlift$ is well
  formed and thus satisfies \wf{15}.
  By \wfh{15} on $\Sigma \to^* \Sigma'$ each spawned thread is assigned
  to a single core, by \wfh{8} each thread appears only on a single set
  of threads, and by \wfh{9} the contents of the object cache and the
  write buffer of each core are altered only by that core.
  As a result, since \wf{15} holds by $\rlift$ it is also true for the
  whole transition, since the core performing the invalidation is the
  only that can alter the object cache and the write buffer and it
  cannot perform an invalidation action in parallel with itself
  (first premise of $\rparg$).\\

  \wf{16}: By induction on the eighth and ninth premise of $\rparg$,
  every thread steps through the $\rlift$, $\rspawn$, or $\rmigrate$,
  and specifically reads step through $\rlift$.
  According to \cref{l:lift} every step performed by $\rlift$ is well
  formed and thus satisfies \wf{15}.
  By \wfh{16} on $\Sigma \to^* \Sigma'$ each spawned thread is assigned
  to a single core, by \wfh{8} each thread appears only on a single set
  of threads, and by \wfh{9} the contents of the object cache and the
  write buffer of each core are altered only by that core.
  As a result, since \wf{16} holds by $\rlift$ it is also true for the
  whole transition, since the core performing the read is the only that
  can alter the object cache and the write buffer and it cannot perform
  a write-back action in parallel with itself (first premise of $\rparg$).\\

  \wf{17}: See \cref{l:wf17}.\\

  \wf{18}: See \cref{l:wf18}.\\

  \wf{19}: See \cref{l:wf19}.\\

  \wf{20}: By \wf{20} on $\Sigma \to^* \Sigma'$ we know that the
  happens-before order between two writes is consistent with the
  happens-before order of their write-backs.
  As a result we only need to examine new write-back actions.
  By induction on the eighth premise of $\rparg$, only one core may
  modify the heap.
  As a result there can only be one write-back in the transition at
  hand, which cannot break the happens before order consistency.
  As a result \wf{20} is satisfied.\\

  \wfe{1}: By induction on the eighth and ninth premise of $\rparg$,
  every thread steps through the $\rlift$, $\rspawn$, or $\rmigrate$,
  and specifically reads step through $\rlift$.
  According to \cref{l:lift} every step performed by $\rlift$ is well
  formed and thus satisfies \wfe{1}.
  That is, there is a corresponding fetch action between thread
  migration and every read action performed by the core that the
  corresponding thread migrated to.
  As a result, only the parallel evaluation of a migration and a read
  action could break this rule.
  However, since those two actions should be performed by the same
  thread this is not possible.
  By \wfh{7} on $\Sigma \to^* \Sigma'$ each spawned thread is assigned
  to a single core, and by \wfh{8} each thread appears only on a single
  set of threads, and by \wfh{9} the contents of the object cache and
  the write buffer of each core are altered only by that core.
  As a result, since \wfe{1} holds by $\rlift$ it is also true for the
  whole transition, since the core performing the read is the only that
  can step the thread at hand and it cannot perform another action in
  parallel with itself (first premise of $\rparg$).\\

  \wfe{2}: By induction on the eighth and ninth premise of $\rparg$,
  every thread steps through the $\rlift$, $\rspawn$, or $\rmigrate$,
  and specifically migrations step through $\rmigrate$.
  \wfe{2} is satisfied by the premises of $\rmigrate$.
  That is, at migration actions there are no dirty data at the
  \emph{old} core, in the two transitions in isolation.
  As a result, only the parallel evaluation of a migration and a write
  action at the \emph{old} core could break this rule.
  However, since those two actions should be performed by the same
  thread this is not possible.
  By \wfh{7} on $\Sigma \to^* \Sigma'$ each spawned thread is assigned
  to a single core, and by \wfh{8} each thread appears only on a single
  set of threads, and by \wfh{9} the contents of the object cache and
  the write buffer of each core are altered only by that core.
  As a result, since \wfe{2} holds by $\rmigrate$ it is also true for the
  whole transition, since the core performing the migration is the only
  that can step the thread at hand and it cannot perform another action
  in parallel with itself (first premise of $\rparg$).\\

  \wfh{1}: By induction on the eighth and ninth premise of $\rparg$,
  every thread steps through the $\rlift$, $\rspawn$, or $\rmigrate$,
  and specifically write-backs step through $\rlift$.
  According to \cref{l:lift} every step performed by $\rlift$ is well
  formed and thus satisfies \wfh{1}.
  Since $\Sigma \to^* \Sigma'$ is well-formed we also know that it
  satisfies \wfh{1} as well.
  As a result we only need to examine new write-back actions.
  By induction on the eighth premise of $\rparg$, only one core may
  modify the heap, thus there can only be one write-back in the
  transition at hand.
  As a result \wfh{1} is satisfied\\

  \wfh{2}: By induction on the eighth and ninth premise of $\rparg$,
  every thread steps through the $\rlift$, $\rspawn$, or $\rmigrate$,
  and specifically fetches and write-backs step through $\rlift$.
  According to \cref{l:lift} every step performed by $\rlift$ is well
  formed and thus satisfies \wfh{2}.
  By \wfh{7} on $\Sigma \to^* \Sigma'$ each spawned thread is assigned
  to a single core, and by \wfh{8} each thread appears only on a single
  set of threads, and by \wfh{9} the contents of the object cache and
  the write buffer of each core are altered only by that core.
  As a result, since \wfh{2} is true for the single step it is also true for
  the whole transition, since the core performing the fetch or
  write-back is the only that can modify the object cache and it cannot
  perform another action in parallel with itself (first premise of
  $\rparg$).\\

  \wfh{3}: By induction on the eighth and ninth premise of $\rparg$,
  every thread steps through the $\rlift$, $\rspawn$, or $\rmigrate$,
  and specifically writes step through $\rlift$.
  According to \cref{l:lift} every step performed by $\rlift$ is well
  formed and thus satisfies \wfh{3}.
  By \wfh{7} on $\Sigma \to^* \Sigma'$ each spawned thread is assigned
  to a single core, and by \wfh{8} each thread appears only on a single
  set of threads, and by \wfh{9} the contents of the object cache and
  the write buffer of each core are altered only by that core.
  As a result, since \wfh{3} is true for the single step it is also true for
  the whole transition, since the core performing the write is the only
  that can modify the write buffer and it cannot perform another action
  in parallel with itself (first premise of $\rparg$).\\

  \wfh{4}: By induction on the eighth and ninth premise of $\rparg$,
  every thread steps through the $\rlift$, $\rspawn$, or $\rmigrate$.
  According to \cref{l:lift} every step performed by $\rlift$ is well
  formed and thus satisfies \wfh{4}.
  $\rspawn$ and $\rmigrate$ are of no interest since they do not alter
  the object cache.
  As a result, \wfh{4} is also satisfied in the whole transition since
  it is satisfied by every step in the transition.\\

  \wfh{5}: By induction on the eighth and ninth premise of $\rparg$,
  every thread steps through the $\rlift$, $\rspawn$, or $\rmigrate$.
  According to \cref{l:lift} every step performed by $\rlift$ is well
  formed and thus satisfies \wfh{4}.
  $\rspawn$ and $\rmigrate$ are of no interest since they do not alter
  the values of any variables.
  As a result, \wfh{5} is also satisfied in the whole transition since
  it is satisfied by every step in the transition.\\

  \wfh{6}: Since $\Sigma \to^* \Sigma'$ is well-formed we also know that
  it satisfies \wfh{6} as well.
  As a result we only need to examine new volatile writes.
  By induction on the eighth premise of $\rparg$, only one core may
  modify the heap, thus there can only be one volatile write in the
  transition at hand.
  As a result \wfh{6} is satisfied\\

  \wfh{7} and \wfh{8}: Since $\Sigma \to^* \Sigma'$ is well-formed we
  also know that it satisfies \wfh{7} and \wfh{8} as well.
  As a result we only need to examine new spawns.
  By induction on the eighth premise of $\rparg$, only one core may
  modify the heap, thus there can only be one spawn in the transition at
  hand.
  By induction on the eighth premise of $\rparg$, we see that a spawn
  can only step through $\rspawn$.
  The spawned thread is assigned to a single core and the old thread
  remains assigned to its core.
  The spawned thread also gets marked as spawned in order to forbid
  future re-spawns of the same thread (first and second premise of
  $\rspawn$).
  As a result, \wfh{7} and \wfh{8} are satisfied, since they are also
  satisfied in $\Sigma \to^* \Sigma'$.\\

  \wfh{9}: Since \wfh{9} is satisfied by $\Sigma to^* \Sigma'$ we
  examine how the current transition alters object caches and write
  buffers.
  By induction on the eighth and ninth premise of $\rparg$, we see that
  all actions altering the object caches and write buffers are evaluated
  by $\rlift$.
  According to \cref{l:lift} every step performed by $\rlift$ is well
  formed and thus satisfies \wfh{9}.
  Since \wfh{9} is satisfied by $\rlift$, it is also true for the whole
  transition, since the object caches and write buffers are disjoint.\\

\end{proof}

\end{document}